\documentclass[10pt]{article}

\usepackage{amscd}
\usepackage{wrapfig}
\usepackage{subfig}
\usepackage{float}
\usepackage{color,soul}
\usepackage{listings}
\usepackage{caption}
\usepackage{bbm}
\usepackage{pict2e}
\usepackage{diagbox}

\usepackage{amsmath,amssymb,amsfonts,amsthm}
\usepackage{moreverb,rotating,graphics}
\usepackage[T1]{fontenc}
\usepackage{epsfig}
\usepackage[francais,english]{babel}
\usepackage{dsfont}

\newtheorem{theorem}{Theorem}

\newtheorem{remark}[theorem]{Remark}
\newtheorem{lemma}[theorem]{Lemma}

\def\N{{\mathbb N}}
\def\Z{{\mathbb Z}}
\def\R{{\mathbb R}}
\def\C{{\mathbb C}}

\def\1{{\mathds{1}}}

\newcommand \dps{\displaystyle }

\newcommand{\dd}{{\mathrm{d}}}
\newcommand{\Bthree}{\mathcal{B}^{(3)}}
\newcommand{\Bfour}{\mathcal{B}^{(4)}}

\newcommand{\Bfive}{\mathcal{B}^{(5)}}

\newcommand{\Ve}{V_{\epsilon}}
\newcommand{\pe}{\psi_{\epsilon}}
\newcommand{\me}{m_\epsilon}
\newcommand{\Te}{\mathcal{T}_\epsilon}

\newcommand{\omitit}[1]{}

\newcommand{\half}{{\textstyle{\frac{1}{2}}}}

\newcommand{\set}[2]{\left\lbrace #1 \, : \, #2 \right\rbrace}
\newcommand{\bra}[1]{\langle\kern-.15em #1 |}
\newcommand{\ket}[1]{| #1 \kern+.05em\rangle}
\newcommand{\inner}[2]{\big\langle #1,#2 \big\rangle}
\newcommand{\biginner}[2]{\big\langle #1,#2 \big\rangle}

\newcommand{\norm}[1]{\left\|{#1}\right\|}
\newcommand{\biggnorm}[1]{\bigg\|{#1}\bigg\|}

\newcommand{\litlo}[1]{o\left(#1\right)}

\newcommand{\rr}{{\bf r}}

\newcommand{\Reyuls}{\mathbb{R}}

\newcommand{\Rtre}{{\Reyuls^3}}

\newcommand{\Rsix}{{\Reyuls^3 \times \Reyuls^3}}

\newcommand{\gsop}{H_0}

\newcommand{\se}{{\cal S}}
\newcommand{\tse}{{\cal T}_\epsilon}

\newcommand{\ao}{{\cal C}}
\newcommand{\te}{{\tau_\epsilon}}

\newcommand{\id}{{\cal I}}

\newcommand{\eps}{{\epsilon}}

\def\be{{\mathbf e}}
\def\br{{\mathbf r}}

\def\bR{{\mathbf r}}

\newcommand{\beec}{{B_{\kern-.10em K}^{\kern+.05em c}}}

\title{Van der Waals interactions between~two~hydrogen~atoms: \\
The next orders}
\date{\today}
\author{Eric Canc\`es\footnotemark[1]  \and Rafa\"el Coyaud\footnotemark[1] \and
L. Ridgway Scott\footnotemark[3]
}

\begin{document}

\maketitle

\begin{abstract}
We extend a method (E. Canc\`es and L.R. Scott, SIAM J. Math. Anal., 50, 2018, 381--410) to compute more terms in the asymptotic expansion of the van der Waals attraction between two hydrogen atoms. These terms are obtained by solving a set of modified Slater--Kirkwood partial differential equations. The accuracy of the method is demonstrated by numerical simulations and comparison with other methods from the literature. It is also shown that the scattering states of the hydrogen atom, that are the states associated with the continuous spectrum of the Hamiltonian, have a major contribution to the C$_6$ coefficient of the van der Waals expansion.
\end{abstract}

\footnotetext[1]{CERMICS, Ecole des Ponts, and Inria Paris, 6 \& 8 avenue Blaise Pascal,
77455 Marne-la-Vall\'ee, France}
\footnotetext[3]{Emeritus, University of Chicago, Chicago, Illinois\ \ 60637, USA}


\section{Introduction}

Van der Waals interactions, first introduced in 1873 to reproduce experimental results on simple gases~\cite{van2004continuity}, have proved to also play an essential role in complex systems in the condensed phase, such as biological molecules~\cite{baldwin2007energetics, roth1996van} and 2D materials~\cite{geim2013van}. The quantum mechanical origin of the dispersive van der Waals interaction has been elucidated by London in the 1930s~\cite{london1937general}. The rigorous mathematical foundations of the van der Waals interaction have been investigated in the pioneering work by Lieb and Thiring~\cite{lieb1986universal}, and later by many authors (see in particular~\cite{ref:AnapolitanosvanderWaals, anapolitanos2020van, anapolitanos2019differentiability, anapolitanos2017long, barbaroux2019van, ref:GeroStudentvanderWaals} and references therein).

In a recent paper~\cite{lrsBIBhi}, a new numerical approach was introduced to compute
the leading order term $-C_6 R^{-6}$ of the
van der Waals interaction between hydrogen atoms separated by a distance $R$.
Here we extend that approach to compute higher order terms $-C_n R^{-n}$, $n>6$.
The coefficients $C_n$ have been computed by various methods.
On the one hand, both~\cite{pauling1935van} and~\cite{Choy2000} apparently failed
to include key components in the computation of $C_{10}$, computing only
one component out of three that we derive here.
On the other hand, our result differs by approximately 200\% and agrees with~\cite{ovsiannikov2005regular}.
One of the objects of this paper is to clarify this discrepancy.

The computation of the expansion coefficients can also be derived through
techniques using polarizabilities~\cite{ovsiannikov2005regular}
which is exact but might involve slightly different numerical computations
than the perturbation method used here. In order to get the right values, one has to use a high enough order of perturbation theory. Computations using up to the second order~\cite{alves2010van, cebim2009high, thakkar1988higher} fail for $C_{12}$, $C_{14}$ and $C_{16}$ (with errors of approximately  1\%, 5\%, and 10\%) for which computations up to the fourth order~\cite{mitroy2005higher} are needed. The third order~\cite{yan1999third} is sufficient for $C_{11}$, $C_{13}$ and $C_{15}$.
Moreover, the polarizabilities method can be derived also for other atoms than hydrogen
as well as for three-body interaction~\cite{cebim2009high}.
A comparison of the numerical results is explored in Section~\ref{sect:comparisonOfResults}.

One can also compute the expansion coefficients using basis states
as in~\cite{forestell2015importance}.
However, this leads to a substantial error even for $C_6$.
The discrepancy observed between the basis states method
and the other methods can be interpreted as
the missing contribution to the energy from the continuous spectrum.

The perturbation method of~\cite{ref:SlaKervdWgas} is remarkable because,
in the case of two hydrogen atoms, the problem splits, for any of the $C_n$ terms,
exactly into terms constituted of an angular factor
and a function of two one-dimensional variables (the underlying problem is six-dimensional).
The first term in this expansion has been examined in \cite{lrsBIBhi} and gave
a value of $C_6$ agreeing with~\cite{ovsiannikov2005regular}.
This article extends this analysis and allows computation of all $C_n$.
The linearity and the nature of the angular parts allows treatment of these
problems separately in a way analogous to the first term of the expansion.
Although the partial differential equations (PDE) defining the functions of these two
variables are not solvable in closed form, they are nevertheless easily solved by
numerical techniques.

In Section~\ref{sec:H2}, we present an extended and modified version of Slater and Kirkwood's
derivation, in order to manipulate more suitable family of PDEs for theoretical
analysis and numerical simulation. These modified Slater--Kirkwood PDEs are well posed at all orders and,
when their unique solutions are multiplied by their respective angular factor,
the resulting function, after summation of the terms, solves the triangular systems of six-dimensional PDEs originating from the Rayleigh--Schr\"odinger expansion.
We finally check that the so-obtained perturbation series are asymptotic expansions of the ground state energy and wave function (after applying some ``almost unitary'' transform) of the hydrogen molecule in the dissociation limit.
In Section~\ref{sec:Cn}, we use a Laguerre approximation~\cite[Section 7.3]{lrsBIBih} to compute coefficients
up to $C_{19}$, given that $C_6$ has been computed in~\cite{lrsBIBhi}.
Our approach also allows us to evaluate the respective contributions of the bound and scattering
states of the Hamiltonian of the hydrogen atom to the $C_6$ coefficient of the van der Waals interaction.
Numerical simulations show that the terms in the sum-over-states expansion coupling two bound states only contribute to about 60\%.
The mathematical proofs are gathered in Section~\ref{sec:proofs}. Lastly, some useful results on the multipolar expansion of the hydrogen molecule electrostatic potential in the dissociation limit and on the Wigner $(2n+1)$ rule used in the computations are provided in the Appendix.

\section{The hydrogen molecule in the dissociation limit}
\label{sec:H2}

As usual in atomic and molecular physics, we work in atomic units: $\hbar=1$ (reduced Planck constant), $e=1$ (elementary charge), $m_e=1$ (mass of the electron), $\epsilon_0=1/(4\pi)$ (dielectric permittivity of the vacuum). The length unit is the bohr (about $0.529$ {\AA}ngstroms) 
and the energy unit is the hartree (about $4.36\times 10^{-18}$ Joules).

We study the Born-Oppenheimer approximation of a system of two hydrogen atoms,
consisting of two classical point-like nuclei of charge $1$ and two quantum
electrons of mass $1$ and charge $-1$. Let $\bR_1$ and $\bR_2$ be the positions
in $\R^3$ of the two electrons, in a cartesian frame whose origin is the center of mass of the nuclei.
We denote by $\be$ the unit vector pointing in the direction from one hydrogen atom to the other,
and by $R$ the distance between the two nuclei. We introduce the parameter $\epsilon=R^{-1}$ and derive expansions in $\epsilon$ of the ground state energy and wave function. Note that in~\cite{lrsBIBhi}, we use instead $\epsilon=R^{-1/3}$. The latter is well-suited to compute the lower-order coefficient $C_6$, but the change of variable $\epsilon=R^{-1}$ is more convenient to compute all the terms of the expansion.

Since the ground state of the hydrogen molecule
is a singlet spin state~\cite{helgaker2013molecular}, its wave function can be written as
\begin{equation}
\label{eqn:spings}
\psi_\epsilon(\bR_1,\bR_2) \frac{\ket{\uparrow\downarrow}-\ket{\downarrow\uparrow}}{\sqrt{2}},
\end{equation}
where $\psi_\epsilon > 0$ is the $L^2$-normalized ground state of the spin-less six-dimensional Schr\"odinger equation
\begin{equation}
\label{eqn:sixdscheq}
H_\eps \psi_\eps=\lambda_\eps \psi_\eps, \quad
 \| \psi_\eps \|_{L^2(\R^3 \times \R^3)}=1,
\end{equation}
where for $\eps > 0$, the Hamiltonian $H_\eps$ is the self-adjoint operator on $L^2(\R^3 \times \R^3)$ with domain $H^2(\R^3 \times \R^3)$ defined by
$$
H_\eps = -\frac 12 \Delta_{\bR_1}-\frac 12 \Delta_{\bR_2} -\frac{1}{|{\bR_1}-(2\eps)^{-1} \be|}- \frac{1}{|{\bR_2}-(2\eps)^{-1}  \be|}
-\frac{1}{|{\bR_1}+(2\eps)^{-1}  \be|}- \frac{1}{|{\bR_2}+(2\eps)^{-1}  \be|}
+\frac{1}{|{\bR_1}-{\bR_2}|}+\eps,
$$
where $\Delta_{\bR_k}$ is the Laplace operator with respect to the variables $\bR_k \in \R^3$. The first two terms of $H_\eps$ model the kinetic energy of the electrons, the next four terms the electrostatic attraction between nuclei and electrons, and the last two terms the electrostatic repulsion between, respectively, electrons and nuclei. The ground state of $H_\eps$ is symmetric ($\psi_\eps(\bR_1,\bR_2)=\psi_\eps(\bR_2,\bR_1)$) so that the wave function defined by \eqref{eqn:spings} does satisfy the Pauli principle (the anti-symmetry is entirely carried by the spin component). It is well-known \cite{lrsBIBhi} that
$$
\lambda_\eps = -1 - C_6 \eps^6 + \litlo{\eps^6}.
$$
The computation of $\lambda_\eps$ (and $\psi_\eps$) to higher order by a modified version of the Slater--Kirkwood approach, is the subject of this article.

\subsection{Perturbation expansion}

The first step is to make a change of coordinates. Introducing the translation operator
$$
\te f(\bR_1,\bR_2)
=f(\bR_1+(2\eps)^{-1} \be, \bR_2- (2\eps)^{-1} \be)
=f(\bR_1+\half{R}\be, \bR_2- \half{R} \be), \quad R=\epsilon^{-1},
$$
the swapping operator $\ao$ and the symmetrization operator $\se$ defined by
$$
\ao \phi(\bR_1,\bR_2) =\phi(\bR_2,\bR_1), \quad \quad
\se =\frac{1}{\sqrt{2}}(\id+\ao),
$$
where $\id$ denotes the identity operator,
as well as the ``asymptotically unitary'' operator
\begin{equation} \label{eqn:simmoperaursop}
\tse=\se\te.
\end{equation}
It is shown in~\cite{lrsBIBhi} that
\begin{equation} \label{eqn:pointoftransop}
H_\epsilon\tse=\tse\big(\gsop + V_\epsilon\big),
\end{equation}
where $\gsop$ is the reference non-interacting Hamiltonian
$$
\gsop=-\frac 12 \Delta_{\br_1}  - \frac{1}{|\br_1|}
      - \frac 12 \Delta_{\br_2} - \frac{1}{|\br_2|},
$$
and $V_\epsilon$ the correlation potential
\begin{equation}
\label{eqn:hydropot}
V_\epsilon({\br_1},{\br_2})=- \frac{1}{|{\br_1}-\epsilon^{-1}\be|}- \frac{1}{|{\br_2}+\epsilon^{-1}\be|}+\frac{1}{|{\br_1}-{\br_2}-\epsilon^{-1}\be|}+\epsilon.
\end{equation}
The linear operator $\tse$ is  ``asymptotically unitary''  in the sense that for all $f,g \in L^2(\R^3 \times \R^3)$,
$$
\inner{\tse f}{\tse g}= \inner{f}{g} + \inner{\ao f}{\tau_{\eps/2} g} \mathop{\longrightarrow}_{\eps \to 0}  \inner{f}{g}.
$$
It follows from \eqref{eqn:pointoftransop} that if $(\lambda,\phi)$ is a normalized eigenstate of $H_0+V_\eps$, that is $(\lambda,\phi)$ satisfies
$$
(H_0+V_\eps) \phi = \lambda \phi, \quad \|\phi\|_{L^2(\R^3 \times \R^3)} = 1,
$$
then
$$
H_\eps \tse\phi = \lambda \tse\phi.
$$
In addition, we know from Zhislin's theorem~\cite{lrsBIBhi,zhislin1960discussion} that both $H_\eps$ and $H_0+V_\eps$ have ground states, that their ground state eigenvalues are non-degenerate, and that their ground state wave functions are (up to replacing them by their opposites) positive everywhere in $\R^{3 \times 3}$. Since $\tse$ preserves positivity, we infer that $H_\eps$ and $H_0+V_\eps$ share the same ground state eigenvalue $\lambda_\eps$ and that if $\phi_\eps$ is the normalized positive ground state wave function of $H_0+V_\eps$, then $\psi_\eps := \tse\phi_\eps/\|\tse\phi_\eps\|_{L^2(\R^3 \times \R^3)}$ is the normalized positive ground state wave function of $H_\eps$.

The next step is to construct for $\eps > 0$ small enough the ground state $(\lambda_\eps,\phi_\eps)$ of $H_0+V_\eps$ by the Rayleigh--Schr\"odinger perturbation method from the explicit ground state
\begin{equation}
\label{eqn:limsighr}
 \lambda_0=-1, \qquad  \phi_0({\br_1},{\br_2})= \pi^{-1} e^{-(|{\br_1}|+|{\br_2}|)},
\end{equation}
of $H_0$. Using a multipolar expansion, we have
\begin{equation}\label{eqn:exp_Veps}
\Ve(\mathbf{r}_1, \mathbf{r}_2) = \sum_{n=3}^{+\infty} \epsilon^{n}
       \mathcal{B}^{(n)}(\mathbf{r}_1,\mathbf{r}_2),
\end{equation}
where homogeneous polynomial functions $\mathcal{B}^{(n)}$, $n \ge 3$ are specified below (see equation \eqref{eqnBic}), the convergence of the series being uniform on every compact subset of $\R^3 \times \R^3$. Assuming that $\lambda_\eps$ and $\phi_\eps$ can be Taylor expanded as
\begin{equation}\label{eqn:exp_lambda_phi}
\lambda_\eps = \lambda_0 - \sum_{n=1}^{+\infty} C_n \eps^n \quad \mbox{and} \quad \phi_\eps = \sum_{n=0}^{+\infty} \eps^n \phi_n, \qquad \mbox{(formal expansions)}
\end{equation}
(we use the standard historical notation $-C_n$ instead of $\lambda_n$ for the coefficients of the eigenvalue $\lambda_\eps$)
inserting these expansions  in the equations $(H_0+V_\eps) \phi_\eps = \lambda_\eps \phi_\eps$,  $\|\phi_\eps\|_{L^2(\R^3 \times \R^3)} = 1$, and identifying the terms of order $n$ in $\eps$, we obtain a triangular system of linear elliptic equations (Rayleigh--Schr\"odinger equations). The well-posedness of this system is given by the following lemma, whose proof is postponed until Section~\ref{sec:proof_lem1}.

\begin{lemma}\label{lem:PsiNEqnAnyOrder} The triangular system
  \begin{align}
\forall n \ge 1, \qquad     (H_0 - \lambda_0) \phi_n &= - \sum_{k = 3}^n \mathcal{B}^{(k)} \phi_{n-k}
      - \sum_{k = 1}^n C_k \phi_{n-k}, \label{eqn:PsiNEqnAnyOrderH0}\\
    \biginner{\phi_0}{\phi_n} &= - \frac 12 \sum_{k = 1}^{n-1}
      \biginner{\phi_k}{\phi_{n-k}},\label{eqn:PsiNEqnAnyOrderPsi0}
  \end{align}
where we use the convention $\sum_{k = m}^{n} \dots=0$ if $m > n$, has a unique solution $((C_n,\phi_n))_{n \in \N^\ast}$ in $(\R \times H^2(\R^3 \times \R^3))^{\N^*}$.
In particular, we have $(C_1,\phi_1)=(C_2,\phi_2)=0$ and $C_3=C_4=C_5=0$. In addition, the functions $\phi_n$ are real-valued.
\end{lemma}

Note that $(C_1,\phi_1)=(C_2,\phi_2)=0$ directly follows from the fact that the first non-vanishing term in the expansion~\eqref{eqn:exp_Veps} of $V_\eps$ is $\eps^3\Bthree$.
The formal expansions \eqref{eqn:exp_lambda_phi} are in fact asymptotic expansions as established in the following theorem.
Its proof is provided in Section \ref{sec:proof_lem1}.

\begin{theorem}
\label{thm:treophmole}
Let $\psi_\epsilon\in H^2(\Rsix)$ be the positive $L^2(\Rsix)$-normalized ground state
of $H_\epsilon$ and $\lambda_\epsilon$ the associated ground-state energy:
\begin{equation}\label{eqn:hydropair}
H_\eps \psi_\eps=\lambda_\eps\psi_\eps, \quad \|\psi_\eps\|_{L^2(\R^3 \times \R^3)}=1, \quad \psi_\eps > 0 \mbox{ a.e. on } \R^3 \times \R^3.
\end{equation}
Let $(\phi_0,\lambda_0)$ be as in \eqref{eqn:limsighr}, $((C_n,\phi_n))_{n \in \N^\ast}$ the unique solution  of \eqref{eqn:PsiNEqnAnyOrderH0} in
$(\R \times H^2(\R^3 \times \R^3))_{n \in \N^\ast}$, and $\tse$ the ``almost unitary'' symmetrization operator defined in \eqref{eqn:simmoperaursop}.
Then, for all $n\in \N$, there exists $\epsilon_n > 0$ and $K_n \in \R_+$ such that for all $0 < \eps \le \eps_n$,
\begin{align} \label{eqn:finalestionsh}
\left\| \psi_\epsilon- \psi_\epsilon^{(n)} \right\|_{H^2(\R^3 \times \R^3)}
&\leq K_n \epsilon^{n+1},  \quad
\left|\lambda_\epsilon-\lambda_\epsilon^{(n)} \right| \leq K_n \epsilon^{n+1},  \quad
\left|\lambda_\epsilon-\mu_\epsilon^{(n)} \right| \leq K_n \epsilon^{2(n+1)},
\end{align}
where
$$
\psi_\epsilon^{(n)}:= \frac{\tse\left(\phi_0+ \sum_{k=3}^n \epsilon^{k} \phi_k\right)}{\left\|\tse\left(\phi_0+ \sum_{k=3}^n \epsilon^{k} \phi_k\right)\right\|_{L^2(\Rsix)}},
\quad \lambda_\epsilon^{(n)}:=\lambda_0-\sum_{k=6}^n C_k \epsilon^{k}, \quad \mu_\epsilon^{(n)} = \langle \psi_\epsilon^{(n)}|H_\epsilon| \psi_\epsilon^{(n)} \rangle.
$$
\end{theorem}

Let us point out that in view of the last two bounds in \eqref{eqn:finalestionsh}, the series expansion of $\mu_\epsilon^{(n)}$ in $\epsilon$ up to order $(2n+1)$, which can be computed from the $\phi_k$'s for $0 \le k \le n$, is given by
$$
\mu_\epsilon^{(n)} = \lambda_0 - \sum_{k=6}^{2n+1} C_k \epsilon^k + O(\epsilon^{2n+2}).
$$
Therefore, the knowledge of the $\phi_k$'s up to order $n$ allows one to compute all the $C_k$'s up to order $(2n+1)$ (Wigner's $(2n+1)$ rule).

\begin{remark}[van der Waals forces] It follows from the Hellmann-Feynman theorem that the van der Waals force ${\bf F}_\epsilon$ acting on the nucleus located at $(2\epsilon)^{-1}\be$ is given by
$$
{\bf F}_\epsilon = \int_{\R^3} \frac{(\br-(2\epsilon)^{-1}\be)}{|\br-(2\epsilon)^{-1}\be|^3} \rho_\epsilon(\br) \, d\br \quad \mbox{with} \quad
\rho_\epsilon(\br)= 2 \int_{\R^3} |\psi_\epsilon(\br,\br')|^2 \, d\br' \quad \mbox{(electronic density)}.
$$
Introducing the approximation ${\bf F}_\epsilon^{(n)}$ of ${\bf F}_\epsilon$ computed from $\psi_\epsilon^{(n)}$ as
$$
{\bf F}_\epsilon^{(n)} = \int_{\R^3} \frac{(\br-(2\epsilon)^{-1}\be)}{|\br-(2\epsilon)^{-1}\be|^3} \rho_\epsilon^{(n)}(\br) \, d\br \quad \mbox{with} \quad
\rho_\epsilon^{(n)}(\br)= 2 \int_{\R^3} |\psi_\epsilon^{(n)}(\br,\br')|^2 \, d\br',
$$
we obtain from the Cauchy-Schwarz inequality, the Hardy inequality in $\R^3$, and \eqref{eqn:finalestionsh} that
$$
|{\bf F}_\epsilon-{\bf F}_\epsilon^{(n)}| \le 8 \|\psi_\epsilon-\psi_\epsilon^{(n)}\|_{H^1(\R^3 \times \R^3)} \|\psi_\epsilon+\psi_\epsilon^{(n)}\|_{H^1(\R^3 \times \R^3)}
\le K'_n \epsilon^{n+1}
$$
for some constant $K'_n \in \R_+$ independent of $\epsilon$ and $\epsilon$ small enough. Since ${\bf F}_\epsilon^{(n)}$ can be Taylor expanded at $\epsilon=0$, we obtain that the force ${\bf F}_\epsilon$ satisfies for all $n \ge 6$
$$
{\bf F}_\epsilon = - \left( \sum_{k=6}^n nC_n \epsilon^{n+1} \right) \be + O(\epsilon^{n+1}).
$$
This extends the result ${\bf F}_\epsilon = -6C_6 \epsilon^7\be+O(\epsilon^8)$ proved in~\cite[Theorem 4]{anapolitanos2019differentiability} for {\em any} two atoms with non-degenerate ground states, to arbitrary order in the simple case of two hydrogen atoms.
\end{remark}

\subsection{Computation of the perturbation series}

The coefficients $\mathcal{B}^{(n)}$ are obtained by a classical multipolar expansion, detailed in Appendix~\ref{sect:RSH} for the sake of completeness. Using  spherical coordinates in an orthonormal cartesian basis $(\mathbf{e}_1,\mathbf{e}_2,\mathbf{e}_3)$ of $\R^3$ for which $\mathbf{e}_3=\mathbf{e}$, so that
\begin{equation}\label{eqn:convention}
\begin{split}
 \mathbf{r}_i &= r_i \big(  \sin(\theta_i) \cos(\phi_i) \mathbf{e}_1
               + \sin(\theta_i) \sin(\phi_i) \mathbf{e}_2 + \cos(\theta_i) \mathbf{e} \big), \\
&\quad \cos(\theta_i) = \mathbf{r}_i \cdot \mathbf{e},
\quad \text{and} \quad r_i = | \mathbf{r}_i |, \quad i=1,2,
\end{split}
\end{equation}
it holds that for all $n \ge 3$,
\begin{align} \label{eqnBic}
\mathcal{B}^{(n)}(\mathbf{r}_1, \mathbf{r}_2) &=
    \sum_{(l_1,l_2)\in B_n} \kern-1em r_1^{l_1} r_2^{l_2} \kern-.5em
    \sum_{- \min(l_1, l_2) \leq m \leq \min(l_1, l_2)} \kern-3em
      G_{\rm c}(l_1,l_2, m)  Y_{l_1}^m(\theta_1, \phi_1) {Y}_{l_2}^{-m}(\theta_2, \phi_2), \\
      &=
    \sum_{(l_1,l_2)\in B_n } \kern-1em r_1^{l_1} r_2^{l_2} \kern-.5em
    \sum_{- \min(l_1, l_2) \leq m \leq \min(l_1, l_2)} \kern-3em
      G_{\rm r}(l_1,l_2, m) \mathcal{Y}_{l_1}^m(\theta_1, \phi_1) \mathcal{Y}_{l_2}^m(\theta_2, \phi_2), \label{eqnBi}
\end{align}
where $({Y}_{l}^m)_{l \in \N, \; m=-l,-l+1,\cdots, l-1,l}$ and $(\mathcal{Y}_l^m)_{l \in \N, \; m=-l,-l+1,\cdots, l-1,l}$ are respectively the complex and real spherical harmonics
, and where
\begin{equation} \label{eqn:whatisbee}
B_n=\set{(l_1,l_2)}{l_1 + l_2 = n-1, \; l_1, l_2 \neq 0 }
=\set{(l,n-1-l)}{ 1\leq l\leq n-2}.
\end{equation}
The coefficients $G_{\rm c}(l_1, l_2, m)$ and  $G_{\rm r}(l_1, l_2, m)$ are respectively given by
\begin{align}\label{eq:defG}
  G_{\rm c}(l_1, l_2, m) &:= (-1)^{l_2}
\frac{4\pi(l_1+l_2)!}{\big( (2l_1 +1)(2l_2 +1)(l_1-m)!(l_1+m)!(l_2 -m)!(l_2+m)! \big)^{1/2} }, \\
  G_{\rm r}(l_1, l_2, m) &:= (-1)^{m} G_{\rm c}(l_1, l_2, m). \nonumber
\end{align}
Both expansions \eqref{eqnBic} and \eqref{eqnBi} are useful: \eqref{eqnBic} will be used in the proof of Theorem~\ref{thm:phin} to establish formula \eqref{eqn:psiiform}, which has a simpler and more compact form in the complex spherical harmonics basis. On the other hand,  \eqref{eqnBi}  allows one to work with real-valued functions.

\medskip

One of the main contributions of this article is to show that the functions $\phi_n$, hence the real numbers $\lambda_n$, can be obtained by solving simple 2D linear elliptic boundary value problems on the quadrant
$$
\Omega = \R_+^\ast \times \R_+^\ast.
$$
For each angular momentum quantum number $l \in \N$, we denote by
\begin{equation}\label{eqn:kappaAlphaDef}
\kappa_l (r) = \frac{l ( l+1)}{2r^2} - \frac{1}{r}- \frac{1}{2} \lambda_0 =  \frac{l ( l + 1)}{2r^2} - \frac{1}{r}+ \frac{1}{2},
\end{equation}
and we consider the boundary value problem: given $f \in L^2(\Omega)$
\begin{equation}\label{eqn:TMain}
\left\{ \begin{array}{l} \dps \mbox{find} T \in {H_0^1}(\Omega) \mbox{ such that} \\
\dps - \frac{1}{2} \Delta T ( r_1, r_2) + \left( \kappa_{l_1} ( r_1) + \kappa_{l_2} (r_2) \right) T = f(r_1, r_2) \quad \mbox{ in } {\cal D}'(\Omega).
\end{array} \right.
\end{equation}
It follows from classical results on the radial operator $- \frac{1}{2} \frac{d^2}{dr^2}+ \kappa_{l}$ on $L^2(0,+\infty)$ with form domain $H^1_0(0,+\infty)$ encountered in the study of the hydrogen atom (see Section~\ref{sec:proof_lemRS} for details) that for all $l_1,\ l_2 \in \mathbb{N}$, $(l_1, l_2) \neq (0, 0)$, the problem \eqref{eqn:TMain} is well posed in $H_0^1(\Omega)$. For $l_1=l_2=0$, this problem is well-posed in
$$
\widetilde{H_0^1}(\Omega)=\set{v\in H_0^1(\Omega)}{\int_\Omega v(r_1,r_2) e^{-r_1-r_2}\, r_1 r_2\,\dd r_1 \dd r_2 =0},
$$
provided that the compatibility condition
\begin{equation}\label{eqn:effcond}
\int_\Omega f(r_1,r_2) e^{-r_1-r_2}\, r_1 r_2\,\dd r_1 \dd r_2 =0
\end{equation}
is fulfilled. Problem \eqref{eqn:TMain} is useful to solve the Rayleigh--Schr\"odinger system \eqref{eqn:PsiNEqnAnyOrderH0}-\eqref{eqn:PsiNEqnAnyOrderPsi0} thanks to the following lemma, proved in Section~\ref{sec:proof_lemRS}. We denote by
$$
\phi_0^\perp := \set{\psi \in L^2(\R^3 \times \R^3)}{ \inner{\phi_0}{\psi}=0}.
$$
Note that the condition \eqref{eqn:effcond} is equivalent to $\inner{\phi_0}{\frac{f(r_1,r_2)}{r_1r_2}}=0$.

\begin{lemma} \label{lem:f_to_T}
Let $l_1,l_2 \in \N$, $m_1,m_2 \in \Z$ such that $-l_j \le m_j \le l_j$ for $j=1, 2$, and $f \in L^2(\Omega)$. Consider the problem of finding $\psi \in H^2(\R^3 \times \R^3) \cap \phi_0^\perp$ solution to the equation
\begin{equation} \label{eq:eq_res}
(H_0-\lambda_0)\psi = F \quad \mbox{with} \quad F:=\frac{f(r_1,r_2)}{r_1r_2}  {Y}_{l_1}^{m_1}(\theta_1, \phi_1)  {Y}_{l_2}^{m_2}(\theta_2, \phi_2) .
\end{equation}
\begin{enumerate}
\item If $(l_1,l_2) \neq (0,0)$, then the unique solution to \eqref{eq:eq_res} in $H^2(\R^3 \times \R^3)$ is
\begin{equation}\label{eq:eq_res_psi}
\psi =   \frac{T(r_1,r_2)}{r_1r_2}  {Y}_{l_1}^{m_1}(\theta_1, \phi_1)  {Y}_{l_2}^{m_2}(\theta_2, \phi_2),
\end{equation}
where $T$ is the unique solution to \eqref{eqn:TMain} in $H^1_0(\Omega)$;
\item  If $(l_1,l_2) = (0,0)$, and if the compatibility condition \eqref{eqn:effcond} is satisfied, then the unique solution to \eqref{eq:eq_res} in $H^2(\R^3 \times \R^3) \cap \phi_0^\perp$ is
$$
\psi = \frac{1}{4\pi}  \frac{T(r_1,r_2)}{r_1r_2},
$$
where $T$ is the unique solution to \eqref{eqn:TMain} in $\widetilde H^1_0(\Omega)$.
\end{enumerate}
In addition, if $f$ decays exponentially at infinity, then so does $T$, hence $\psi$, in the following sense: for all $0 \le \alpha < \sqrt{3/8}$,  there exists a constant $C_\alpha \in \R_+$ such that for all $\eta > \alpha$, $l_1,l_2 \in \N$, $m_1,m_2 \in \Z$ such that $-l_j \le m_j \le l_j$ for $j=1,2$, and all $f \in L^2(\Omega)$
\begin{align}
\Vert e^{\alpha (r_1+ r_2) } T \Vert_{H^1 ( \Omega)} &\le C_\alpha \Vert  e^{\eta ( r_1 + r_2)} f \Vert_{L^2 ( \Omega)}, \\
\Vert  e^{\alpha (|\br_1|+ |\br_2|)} \psi \Vert_{L^2 (\Rsix)} &\le C_\alpha \Vert e^{\eta ( |\br_1| + |\br_2|)} F \Vert_{L^2 ( \Rsix)}, \label{eq:boundL2exp} \\
\Vert  e^{\alpha (|\br_1|+ |\br_2|)} \psi \Vert_{H^1 (\Rsix)} &\le C_\alpha (1+4l_1(l_1+1)+4l_2(l_2+1))^{1/2} \Vert e^{\eta ( |\br_1| + |\br_2|)} F \Vert_{L^2 ( \Rsix)}. \label{eq:boundH1exp}
\end{align}
Lastly, if $f$ is real-valued, then so is $T$.
\end{lemma}

The properties of the functions $\phi_n$ upon which our numerical method is based, are collected in the following theorem, proved in Section~\ref{sec:proof_lem1}.

\begin{theorem} \label{thm:phin}
Let $((C_n,\phi_n))_{n \in \N^\ast}$ be the unique solution in
$(\R \times H^2(\R^3 \times \R^3))_{n \in \N^\ast}$ to the Rayleigh--Schr\"odinger system  \eqref{eqn:PsiNEqnAnyOrderH0}. Then, $\phi_1=\phi_2=0$, $C_n=0$ for $1 \le n \le 5$ and for each $n \ge 3$, there exists a positive integer $N_n$ such that
\begin{equation}\label{eqn:psiiform}
 \phi_n = \sum_{(l_1,l_2) \in {\cal L}_n}  \frac{T^{(n)}_{(l_1,l_2)}( r_1, r_2)}{r_1 r_2} \left(  \sum_{m=-\min(l_1,l_2)}^{\min(l_1,l_2)} \alpha^{(n)}_{(l_1,l_2,m)} {Y}_{l_1}^{m} (\theta_1, \phi_1) {Y}_{l_2}^{-m}  ( \theta_2, \phi_2) \right),
\end{equation}
where ${\cal L}_n$ is a finite subset of $\N^2$ with cardinality $N_n<\infty$, where $T^{(n)}_{(l_1,l_2)}$ is the unique solution to \eqref{eqn:TMain} in $H^1(\Omega)$ (or in $\widetilde H^1(\Omega)$ if $l_1=l_2=0$) for $f=f^{(n)}_{(l_1,l_2)}$, where $f^{(n)}_{(l_1,l_2)}$ is a real-valued function that can be computed recursively from the $T^{(n')}_{(l_1',l_2')}$'s, for $n' < n$, and where $\alpha^{(n)}_{(l_1,l_2,m)}$ are real coefficients. Moreover, there exists $\alpha_n > 0$ such that
\begin{align}
\Vert  e^{\alpha_n (r_1+ r_2)} T^{(n)}_{(l_1,l_2)} \Vert_{H^1 ( \Omega)} &< \infty, \\
\Vert e^{\alpha_n (|\br_1|+ |\br_2|)}  \phi_n  \Vert_{H^1 (\Rsix)} & < \infty. \label{eq:exp_dec_phin}
\end{align}
\end{theorem}

The number $N_n=|{\cal L}_n|$ (number of terms in the expansion) for $6 \le n \le 9$
are displayed in Table~\ref{tbl:lsphericalharmonicspsii}, whose construction rules are given in the proof of Theorem~\ref{thm:phin} (see Section~\ref{sec:proof_lem1}).
For $3 \le n\le 5$, ${\cal L}_n=B_n$, where the latter set is defined in \eqref{eqn:whatisbee},
and $N_n=|B_n|=n-2$.
For general $n$, $B_n\subset{\cal L}_n$.
For $n\ge 6$, additional terms appear, as indicated in Table~\ref{tbl:lsphericalharmonicspsii}.
\begin{table}[H]
\begin{center}
  \begin{tabular}{|l|l|l|}
   \hline
   $n$ & $N_n$ & pairs of angular momentum quantum numbers $(l_1,l_2)$
                   in ${\cal L}_n\backslash B_n$  \\
   \hline
    6 &  8   & (0,2;0,2)\\
    7 &  13  &  (0,2;1,3), (1,3;0,2)\\
    8 &  18  &  (0,2;0,2,4), (1,3;1,3), (0,2,4;0,2) \\
    9 &  27  &  (0,2;1,3,5), (1,3;0,2,4), (1,3,5;0,2), (0,2,4;1,3),   (1,3;1,3) \\
   \hline
   \end{tabular}
 \end{center}
 \caption{Additional spherical harmonics appearing in each $\phi_n$ for $6 \le n \le 9$.
$N_n$ is the number of terms in the spherical harmonics expansion~\eqref{eqn:psiiform}. The condensed notation $(l_1,l_1';l_2,l_2')$ (resp. $(l_1,l_1';l_2,l_2',l_2'')$ or $(l_1,l_1',l_1'';l_2,l_2')$) stands for the four (resp. six) pairs $(l_1,l_2)$, $(l_1',l_2)$, $(l_1,l_2')$, etc.}
\label{tbl:lsphericalharmonicspsii}
\end{table}

Table \ref{tbl:lsphericalharmonicspsii}  can be read using the following rule:
for a given $n$, if $(l_1, l_2)$ appears in the corresponding row of the table, then there may exist $m$ such that
  $Y_{l_1}^{m}(\theta_1, \phi_1) Y_{l_2}^{-m}(\theta_2, \phi_2)$ might appear
  with a non-zero coefficient $\alpha^{(n)}_{(l_1,l_2,m)}$ in the spherical harmonics expansion \eqref{eqn:psiiform} of $\phi_n$.
  Conversely, if a given $(l_1, l_2)$ does not appear in the table, then
   $\biginner{\phi_n}{\frac{v(r_1,r_2)}{r_1r_2}Y_{l_1}^{m_1}(\theta_1, \phi_1) Y_{l_2}^{m_2}(\theta_2, \phi_2)} = 0$,
for all $m_1,m_2$ and all $v \in L^2(\Omega)$.
The relative complexity of  Table~\ref{tbl:lsphericalharmonicspsii} is due to fact the first term in the right-hand side of~\eqref{eqn:PsiNEqnAnyOrderH0} is a sum of bilinear terms in  $\mathcal{B}^{(k)}$ and $\phi_{n-k}$. The angular parts of both $\mathcal{B}^{(k)}$ and $\phi_{n-k}$ are finite linear combinations of  angular basis functions ${Y}_{l_1}^{m} \otimes {Y}_{l_2}^{-m}$. When multiplied, they give rise to a still finite but longer linear combination of ${Y}_{l_1}^{m} \otimes {Y}_{l_2}^{-m}$'s (see \eqref{eq:prod_Ylm}).
By contrast, the corresponding table for the $\mathcal{B}^{(n)}$'s is quite simple, since all the rows have the same structure: for all $n \ge 3$, we have
\begin{equation}\label{eqn:beenzterm}
n \quad | \quad n-2  \quad | \quad (k,n-k) \quad \mbox{ for } 1 \le k \le n-2.
\end{equation}

 From $(\phi_k)_{0 \le k \le n}$, we can obtain the coefficients $\lambda_j$ up to $j=2n+1$ using Wigner's $(2n+1)$ rule. Another, more direct, way to compute recursively the $\lambda_n$'s is to take the inner product of $\phi_0$ with each side of~\eqref{eqn:PsiNEqnAnyOrderH0} and use the fact that $\langle \phi_0, (H_0-\lambda_0)\phi_n\rangle=\langle (H_0-\lambda_0)\phi_0, \phi_n\rangle=0$. Since $(C_1,\phi_1)=(C_2,\phi_2)=0$, we thus obtain
\begin{equation}\label{eqn:eigenformula}
     C_n  = - \sum_{k = 3}^{n-3} \biginner{\phi_0}{\mathcal{B}^{(k)} \phi_{n-k}}
     - \sum_{k = 3}^{n-3} C_k \biginner{\phi_0}{\phi_{n-k}},
\end{equation}
where we use the convention $\sum_{k=m}^n ... = 0$ if $m > n$. It follows that $C_3=C_4=C_5=0$.

Using~\eqref{eqnBic},~\eqref{eqn:psiiform} and the orthonormality properties of the complex spherical harmonics, the terms $\inner{\psi_0}{\mathcal{B}^{(k)} \phi_{n-k}}$ in \eqref{eqn:eigenformula} can be written as
\begin{align}
  \inner{\phi_0}{\mathcal{B}^{(k)} \phi_{n-k}}
   &= \inner{\mathcal{B}^{(k)}\phi_0}{ \phi_{n-k}} \nonumber \\
   &= \bigg\langle
    \sum_{(l_1,l_2)\in B_k} \kern-1em r_1^{l_1} r_2^{l_2}
   \sum_{m=-\min(l_1,l_2)}^{\min(l_1,l_2)}
      G_{\rm c}(l_1,l_2, m) {Y}_{l_1}^m(\theta_1, \phi_1) {Y}_{l_2}^{-m}(\theta_2, \phi_2) \pi^{-1} e^{-(r_1+r_2)},  \nonumber  \\
      & \qquad \sum_{(l_1',l_2') \in {\cal L}_{n-k}}  \frac{T^{(n-k)}_{(l_1',l_2')}( r_1, r_2)}{r_1 r_2} \sum_{m'=-\min(l_1',l_2')}^{\min(l_1',l_2')} \alpha^{(n-k)}_{(l_1',l_2',m')}  {Y}_{l_1'}^{m'} (\theta_1, \phi_1) {Y}_{l_2'}^{-m'}  ( \theta_2, \phi_2)  \bigg\rangle  \nonumber   \\
   &= -  \sum_{(l_1,l_2) \in {\cal L}_{n-k}\cap B_k }  \beta^{(n-k)}_{(l_1,l_2)} t^{(n-k)}_{l_1,l_2},
   \label{eq:Bkphin-k}
    \end{align}
    where
    \begin{align} \label{eq:tnl1l2}
    \beta^{(n)}_{(l_1,l_2)}&:= - \pi^{-1}  \sum_{m=-\min(l_1,l_2)}^{\min(l_1,l_2)} \alpha^{(n)}_{(l_1,l_2,m)}  G_{\rm c}(l_1,l_2, m) \\
    t^{(n)}_{(l_1,l_2)}&:= \int_\Omega  r_1^{l_1+1}   r_2^{l_2+1} e^{-(r_1+r_2)} T^{(n)}_{(l_1,l_2)}( r_1, r_2) \, dr_1 \, dr_2,
    \end{align}
    with the convention that $\beta^{(n)}_{(l_1,l_2)}=t^{(n)}_{(l_1,l_2)}=0$ if $(l_1,l_2) \notin {\cal L}_n$. In view of Table~\ref{tbl:lsphericalharmonicspsii}, we see in particular that
since the sum in \eqref{eq:Bkphin-k} is empty
\begin{equation}\label{eqn:whatxerom}
\inner{\phi_0}{\mathcal{B}^{(k)}\phi_n} = 0 \quad\forall \;k,n=3,4,5,\;k\neq n,
\end{equation}
and that many other vanish, e.g.
\begin{equation}\label{eqn:ninergtom}
\inner{\phi_0}{\Bthree \phi_6} = 0,\quad
\inner{\phi_0}{\Bfour \phi_5} = 0,\quad
\inner{\phi_0}{\mathcal{B}^{(5)}\phi_4} = 0,\quad
\inner{\phi_0}{\mathcal{B}^{(6)}\phi_3} = 0.
\end{equation}
Additional pairs $k,n$ can be examined by comparing the sets $B_k$ and ${\cal L}_{n-k}$.

Furthermore, if the chosen numerical method to solve the boundary value problem \eqref{eqn:TMain}
giving the radial function $T^{n-k}_{l_1',l_2'}$ is a Galerkin method using as basis functions of the approximation space tensor products of 1D Laguerre functions (that are, polynomials in $r$ times $e^{-r}$), then the computation of $t^n_{l_1,l_2}$ can be done explicitly, at least for the approximate solution~\cite[Section 7.3]{lrsBIBih}.
Using the fact that
\begin{equation}\label{eqn:groundstateconst}
\phi_0 = 4  e^{-(r_1+ r_2)} {Y}_0^0(\theta_1, \phi_1) {Y}_0^0(\theta_2, \phi_2),
\end{equation}
we then have
\begin{equation}\label{eqn:compAlgoCnComplete}
 \begin{split}
  \inner{\phi_0}{\phi_{n}}
  &= \bigg\langle 4  e^{-(r_1+ r_2)} {Y}_0^0(\theta_1, \phi_1) {Y}_0^0(\theta_2, \phi_2) , \sum_{(l_1',l_2') \in {\cal L}_{n}}  \frac{T^{(n)}_{(l_1',l_2')}( r_1, r_2)}{r_1 r_2} \sum_{m'=-\min(l_1',l_2')}^{\min(l_1',l_2')} \alpha^{(n)}_{(l_1',l_2',m')}  {Y}_{l_1'}^{m'} (\theta_1, \phi_1) {Y}_{l_2'}^{-m'}  ( \theta_2, \phi_2)  \bigg\rangle \\
  &= 4  \alpha^{(n)}_{(0,0,0)} t^{(n)}_{(0,0)}.
    \end{split}
\end{equation}
As a consequence, $\inner{\phi_0}{\phi_{n}}=0$ if $(0,0) \notin {\cal L}_{n}$, so that in particular
\begin{equation}\label{eq:phi0phin}
\inner{\phi_0}{\phi_{3}}=\inner{\phi_0}{\phi_{4}}=\inner{\phi_0}{\phi_{5}}=0.
\end{equation}
Then, $C_n$ can be computed from \eqref{eqn:eigenformula} as
\begin{equation}\label{eqn:CNExpression}
\begin{split}
  C_n = &   \sum_{k=3}^{n-3}  \sum_{\substack{ (l_1,l_2) \in {\cal L}_{n-k} \\l_1 + l_2 = k-1 \\ l_1, l_2 \neq 0 } } \beta^{(n-k)}_{(l_1,l_2)} t^{(n-k)}_{(l_1,l_2)} - 4 \sum_{k=6}^{n-3} C_k  \alpha^{(n-k)}_{(0,0,0)} t^{(n-k)}_{(0,0)}.
\end{split}
\end{equation}

\medskip

\subsection{Practical computation of the lowest order terms}
\label{sec:LOT}

We detail in this section the practical computation of $\phi_3$ (already done in~\cite{lrsBIBhi}), $\phi_4$ and $\phi_5$, as well as $C_n$ for $n \le 11$. Recall that $\phi_1=\phi_2=0$, and $C_n=0$ for $n \le 5$.

\medskip

\noindent
{\bf Computation of $\phi_3$.} We have
\begin{align}
&\Bthree
  =   r_1 r_2
    \left( \sum_{m =  -1}^1 G_{\rm c}(1,1,m)) {Y}_1^m(\theta_1, \phi_1) {Y}_1^{-m}(\theta_2, \phi_2) \right), \label{eq:B3} \\
    & (H_0-\lambda_0)\phi_3 = - \Bthree \phi_0  ,  \label{eq:eqphi3} \\
        & \inner{\phi_0}{\phi_3} =0,  \label{eq:normphi3}
            \end{align}
with $G_{\rm c}(1,1,m)=- \frac{\pi}{3} (8 - 4 |m|)$ and therefore
\begin{align*}
& (H_0-\lambda_0)\phi_3 =  - r_1r_2 e^{-(r_1+r_2)}  \left( \sum_{m =  -1}^1 \pi^{-1}G_{\rm c}(1,1,m) {Y}_1^m(\theta_1, \phi_1) {Y}_1^{-m}(\theta_2, \phi_2) \right), \\
    & \inner{\phi_0}{\phi_3} =0.
    \end{align*}
As a consequence, using Lemma~\ref{lem:f_to_T}, it holds that ${\cal L}_3=\{(1,1)\}$,
\begin{align}\label{eq:phi3_2}
\phi_3 = \frac{T^{(3)}_{(1,1)}(r_1,r_2)}{r_1r_2}  \left( \sum_{m =  -1}^1 \alpha^{(3)}_{(1,1,m)} {Y}_1^m(\theta_1, \phi_1) {Y}_1^{-m}(\theta_2, \phi_2) \right),
\end{align}
where $\alpha^{(3)}_{(1,1,m)}=- \pi^{-1}G_{\rm c}(1,1,m)=- \frac{1}{3} (8 - 4 |m|)$ and where $T^{(3)}_{(1,1)} \in H^1_0(\Omega)$ can be numerically computed by solving the 2D boundary value problem
$$
-\frac 12 \Delta T^{(3)}_{(1,1)} +  \left( \kappa_{1} ( r_1) + \kappa_{1} (r_2) \right) T^{(3)}_{(1,1)} = r_1^2 r_2^2 e^{-(r_1+r_2)} \quad \mbox{ in } \Omega
$$
with homogeneous Dirichlet boundary conditions.

\medskip

\noindent
{\bf Computation of $\phi_4$}. To compute the next order, we first expand $\Bfour$ as
\begin{align*}
\Bfour  =&  r_1 r_2^2 \sum_{m=-1}^1
    G_{\rm c}(1, 2, m) {Y}_1^m (\theta_1, \phi_1) {Y}_2^{-m}(\theta_2, \phi_2)
     + r_1^2 r_2 \sum_{m=-2}^2
      G_{\rm c}(2, 1, m) {Y}_1^m (\theta_1, \phi_1) {Y}_2^{-m}(\theta_2, \phi_2),
\end{align*}
with $G_{\rm c}(1,2,1)=G_{\rm c}(1,2,-1)=4\pi/\sqrt 5$, $G_{\rm c}(1,2,0)=4\pi\sqrt 3/\sqrt 5$, $G_{\rm c}(2,1,m)=-G_{\rm c}(1,2,m)$.
From  \eqref{eqn:PsiNEqnAnyOrderH0}-\eqref{eqn:PsiNEqnAnyOrderPsi0}, we get
\begin{align*}
&(H_0-\lambda_0)\phi_4= -  \mathcal{B}^{(3)} \phi_{1} - \mathcal{B}^{(4)} \phi_{0} , \\
&\inner{\phi_0}{\phi_4} =  0,
\end{align*}
since $\phi_1=\phi_2=0$ and $C_k=0$ for $1 \le k \le 5$. We therefore have ${\cal L}_4=\{(1,2),(2,1)\}$ and
\begin{align*}
\phi_4 =& \frac{T^{(4)}_{(1,2)}(r_1, r_2)}{r_1 r_2} \sum_{m=-1}^1
    \alpha^{(4)}_{(1,2,m)} {Y}_1^m (\theta_1, \phi_1) {Y}_2^{-m}(\theta_2, \phi_2) + \frac{T^{(4)}_{(2,1)}(r_1, r_2)}{r_1 r_2} \sum_{m=-1}^1
    \alpha^{(4)}_{(2,1,m)} {Y}_2^m (\theta_1, \phi_1) {Y}_1^{-m}(\theta_2, \phi_2),
    \end{align*}
where $\alpha^{(4)}_{(l_1,l_2,m)}=- \pi^{-1} G_{\rm c}(l_1,l_2,m)$, $T^{(4)}_{(2,1)} \in H^1_0(\Omega)$ solves
\begin{equation}\label{eqn:T4iMain}
- \frac{1}{2} \Delta_2 T^{(4)}_{(2,1)} ( r_1, r_2) + \left( \kappa_2 ( r_1) + \kappa_1 (r_2) \right) T^{(4)}_{(2,1)} =  r_1^3 r_2^2 e^{-r_1 - r_2 } \quad \mbox{in } \Omega,
\end{equation}
and $T^{(4)}_{(1,2)}(r_1,r_2)=T^{(4)}_{(2,1)}(r_2,r_1)$.
A representation of $T^{(4)}_{(2,1)}$ can be seen in Figure \ref{fig:T4Shape}.
\begin{figure}[t]
\begin{center}
\subfloat[\label{fig:Tshape}]{\includegraphics[width = 0.49\textwidth]{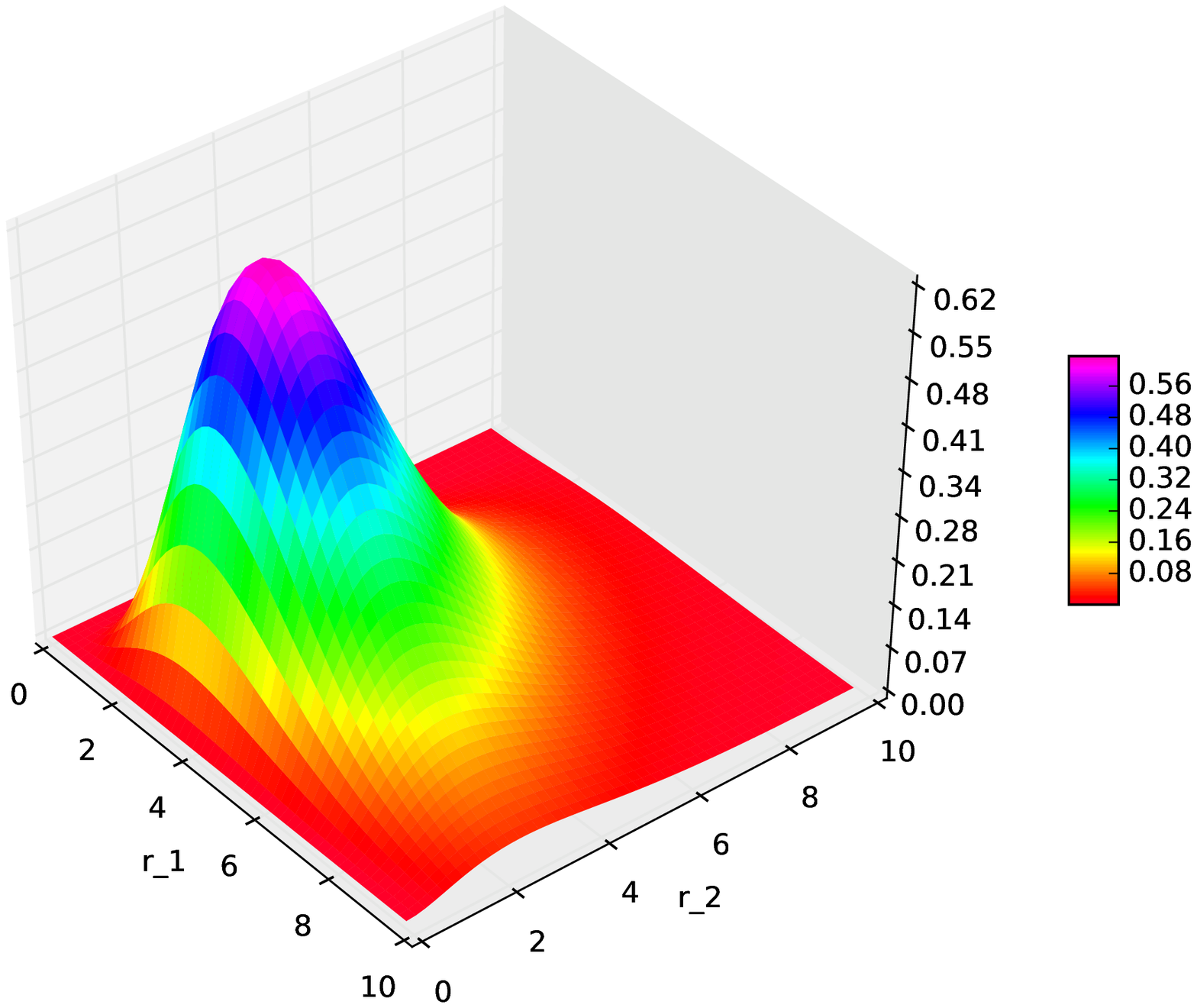}}
\subfloat[\label{fig:Tovrrshape}]{\includegraphics[width = 0.49\textwidth]{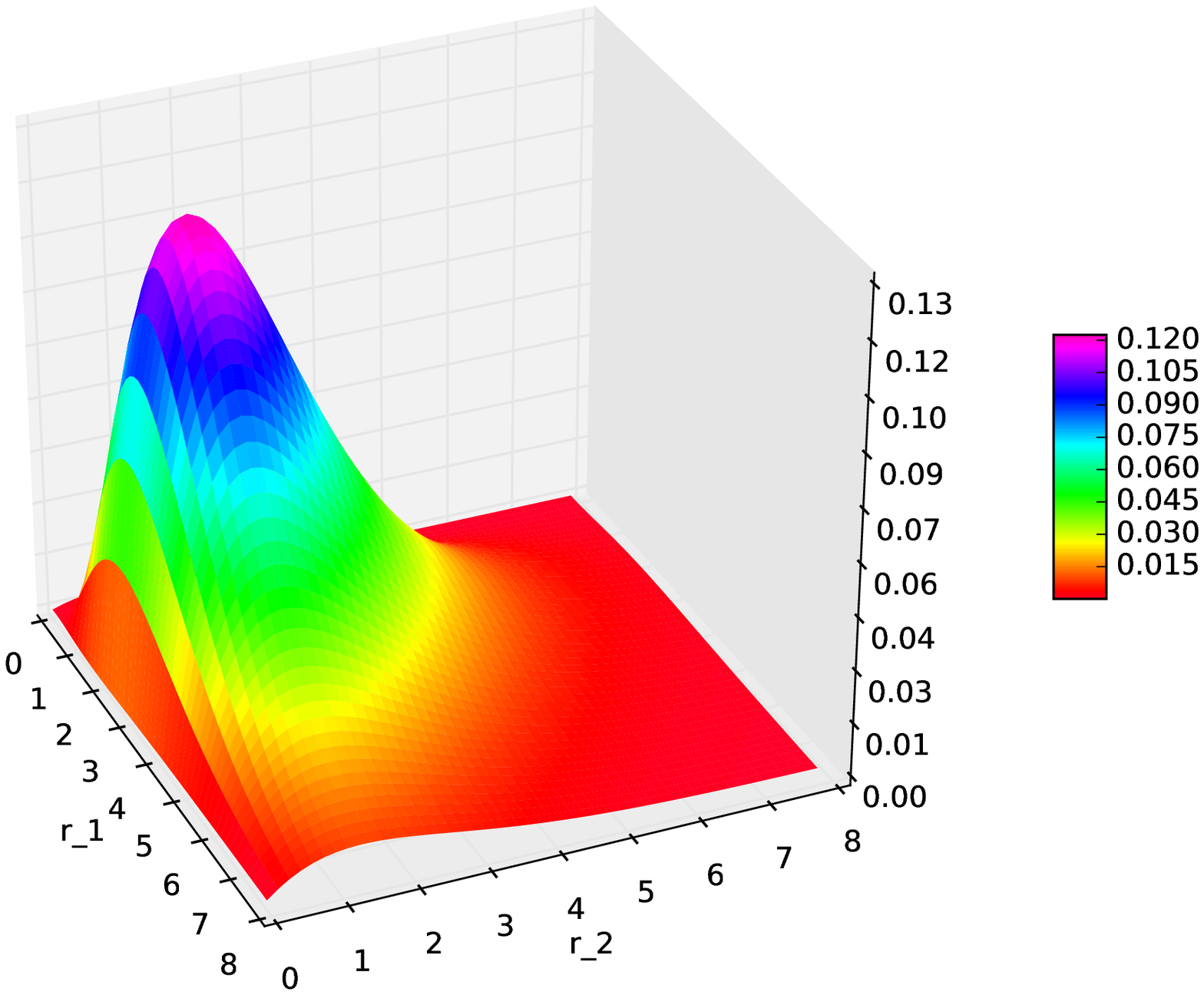}}
\caption{\label{fig:T4Shape} Shape of $T^{(4)}_{(2,1)}$ \protect\subref{fig:Tshape} and
$T^{(4)}_{(2,1)}(r_1, r_2)/(r_1 r_2)$ \protect\subref{fig:Tovrrshape}, using the Laguerre function
approximation scheme~\cite[Section 7.3]{lrsBIBih}.}
\end{center}
\end{figure}

\medskip

\noindent
{\bf Computation of $\phi_5$}. We have
\begin{align*}
  \Bfive =&  r_1 r_2^3 \sum_{m=-1}^1
    G_{\rm c}(1, 3, m) {Y}_1^m (\theta_1, \phi_1) {Y}_3^{-m}(\theta_2, \phi_2)
     + r_1^2 r_2^2 \sum_{m=-2}^2
      G_{\rm c}(2, 2, m) {Y}_2^m (\theta_1, \phi_1) {Y}_2^{-m}(\theta_2, \phi_2)
   \\  &+ r_1^3 r_2^1 \sum_{m=-1}^1
      G_{\rm c}(3, 1, m) {Y}_3^m (\theta_1, \phi_1) {Y}_1^{-m}(\theta_2, \phi_2),
\end{align*}
and
\begin{align*}
&(H_0-\lambda_0)\phi_5= - \mathcal{B}^{(5)} \phi_{0}, \\
&\inner{\phi_0}{\phi_5} =  0,
\end{align*}
since $\phi_1=\phi_2=0$ and $C_k=0$ for $1 \le k \le 5$.  We thus have  ${\cal L}_5=\{(1,3),(2,2),(3,1)\}$ and
\begin{multline}
  \psi^{(5)} = \frac{T^{(5)}_{(1,3)}(r_1,r_2)}{r_1 r_2} \sum_{m=-1}^1
    \alpha^{(5)}_{(1,3,m)}  {Y}_1^m (\theta_1, \phi_1)  {Y}_3^{-m}(\theta_2, \phi_2)
     \\  + \frac{T^{(5)}_{(2,2)}(r_1,r_2)}{r_1 r_2} \sum_{m=-2}^2
    \alpha^{(5)}_{(2,2,m)}  {Y}_2^m (\theta_1, \phi_1)  {Y}_2^{-m}(\theta_2, \phi_2)
     \\  + \frac{T^{(5)}_{(3,1)}(r_1,r_2)}{r_1 r_2}\sum_{m=-1}^1
      \alpha^{(5)}_{(3,1,m)}  {Y}_3^m (\theta_1, \phi_1)  {Y}_1^{-m}(\theta_2, \phi_2) ,
\end{multline}
where $\alpha^{(5)}_{(l_1,l_2,m)}=- \pi^{-1} G_{\rm c}(l_1,l_2,m)$, $T^{(5)}_{(1,3)} \in H^1_0(\Omega)$ solves
\begin{equation}
- \frac{1}{2} \Delta_2 T^{(5)}_{(1,3)} ( r_1, r_2) + \left( \kappa_1 ( r_1) + \kappa_3 (r_2) \right) T^{(5)}_{(1,3)}
= r_1^2 r_2^4 e^{-(r_1 + r_2)},
\end{equation}
$T^{(5)}_{(2,3)}\in H^1_0(\Omega)$ solves
\begin{equation}
- \frac{1}{2} \Delta_2 T^{(5)}_{(2,2)} ( r_1, r_2) + \left( \kappa_2 ( r_1) + \kappa_2 (r_2) \right) T^{(5)}_{(2,2)}
= r_1^3 r_2^3 e^{-(r_1 + r_2) },
\end{equation}
and $T^{(5)}_{(3,1)}(r_1,r_2)=T^{(5)}_{(1,3)}(r_2,r_1)$.

\medskip

\noindent
{\bf Computation of $\lambda_n$ for $6 \le n \le 11$.} From \eqref{eqn:eigenformula} and the fact that $C_n=0$ for $3 \le n \le 5$, we obtain, using  \eqref{eq:Bkphin-k},  \eqref{eq:phi0phin}, \eqref{eqn:CNExpression}, Table~\ref{tbl:lsphericalharmonicspsii}, and the symmetries of the coefficients $\beta^{(n)}_{(l_1,l_2)}$ and $t^{(n)}_{(l_1,l_2)}$,
\begin{align}
C_6&=-\biginner{\phi_0}{\mathcal{B}^{(3)} \phi_{3}} = \beta^{(3)}_{(1,1)} t^{(3)}_{(1,1)}, \label{eq:expC_6}  \\
C_7&=-\biginner{\phi_0}{\mathcal{B}^{(3)} \phi_{4}} - \biginner{\phi_0}{\mathcal{B}^{(4)} \phi_{3}} = 0,   \nonumber  \\
C_8&=- \biginner{\phi_0}{\mathcal{B}^{(3)} \phi_{5}} - \biginner{\phi_0}{\mathcal{B}^{(4)} \phi_{4}} - \biginner{\phi_0}{\mathcal{B}^{(5)} \phi_{3}}
= - \biginner{\phi_0}{\mathcal{B}^{(4)} \phi_{4}}  \nonumber \\
&= \beta^{(4)}_{(1,2)} t^{(4)}_{(1,2)} +  \beta^{(4)}_{(2,1)} t^{(4)}_{(2,1)}= 2 \beta^{(4)}_{(1,2)} t^{(4)}_{(1,2)}, \nonumber \\
C_9&= - \biginner{\phi_0}{\mathcal{B}^{(3)} \phi_{6}} - \biginner{\phi_0}{\mathcal{B}^{(4)} \phi_{5}} - \biginner{\phi_0}{\mathcal{B}^{(5)} \phi_{4}} - \biginner{\phi_0}{\mathcal{B}^{(6)} \phi_{3}} - C_6 \biginner{\phi_0}{\phi_{3}}=0, \nonumber \\
C_{10}&=  - \sum_{k = 3}^{7} \biginner{\phi_0}{\mathcal{B}^{(k)} \phi_{10-k}}
     - \sum_{k = 6}^{7} C_k \biginner{\phi_0}{\phi_{10-k}} = -  \biginner{\phi_0}{\mathcal{B}^{(5)} \phi_{5}}  \nonumber \\
     &=  \beta^{(5)}_{(1,3)} t^{(5)}_{(1,3)} +   \beta^{(5)}_{(2,2)} t^{(5)}_{(2,2)}  +  \beta^{(5)}_{(3,1)} t^{(5)}_{(3,1)} = 2 \beta^{(5)}_{(1,3)} t^{(5)}_{(1,3)} +   \beta^{(5)}_{(2,2)} t^{(5)}_{(2,2)},  \label{eq:expC_10}  \\
C_{11}&= - \sum_{k = 3}^{8} \biginner{\phi_0}{\mathcal{B}^{(k)} \phi_{11-k}}
     - \sum_{k = 6}^{8} C_k \biginner{\phi_0}{\phi_{11-k}}  = -  \biginner{\phi_0}{\mathcal{B}^{(4)} \phi_{7}} -  \biginner{\phi_0}{\mathcal{B}^{(5)} \phi_{6}}  \nonumber \\
     &= \beta^{(7)}_{(1,2)} t^{(7)}_{(1,2)} +   \beta^{(7)}_{(2,1)} t^{(7)}_{(2,1)}  +  \beta^{(6)}_{(2,2)} t^{(6)}_{(2,2)}  \label{eq:expC_11}.
          \end{align}
As $\alpha^{(n)}_{(l_1,l_2,m)}=- \pi^{-1}G(l_1,l_2,m)$ for $n=3,4,5$, $(l_1,l_2) \in {\cal L}_n$ and $-\min(l_1,l_2) \le m \le \min(l_1,l_2)$, we obtain, using \eqref{eq:defG}, that
$$
\big(\alpha^{(n)}_{(l_1,l_2,m)}\big)^2 = \frac{16 \,  ((l_1+l_2)!)^2}{(2l_1 +1)(2l_2 +1)(l_1-m)!(l_1+m)!(l_2 -m)!(l_2+m)!},
$$
and therefore
\begin{align*}
\beta^{(3)}_{(1,1)} &= \sum_{m=-1}^1 (\alpha^{(3)}_{(1,1,m)})^2= \frac{16}{9} +\frac{64}{9}+ \frac{16}{9}= \frac{32}3,  \\
\beta^{(4)}_{(1,2)} &=  \beta^{(4)}_{(2,1)} = \sum_{m=-1}^1 (\alpha^{(4)}_{(1,2,m)})^2 = \frac{16}5 + 3 \times \frac{16}5 + \frac{16}5  =16, \\
\beta^{(5)}_{(1,3)} &=  \beta^{(5)}_{(3,1)} = \sum_{m=-1}^1 (\alpha^{(5)}_{(1,3,m)})^2 =  \frac{64}3, \qquad
\beta^{(5)}_{(2,2)} =  \sum_{m=-2}^2 (\alpha^{(5)}_{(2,2,m)})^2 = \frac{224}{5},
\end{align*}
so that
\begin{align}
&C_6=\frac{32}3 t^{(3)}_{(1,1)}, \quad
C_7=0,  \quad
C_8= 32 t^{(4)}_{(1,2)}, \quad
C_9= 0,  \quad &C_{10} = \frac{128}3 t^{(5)}_{(1,3)} +  \frac{224}{5}t^{(5)}_{(2,2)}.  \label{eq:C6-C10}
\end{align}
It is optimal to use \eqref{eq:C6-C10} to compute $C_6$, $C_8$, $C_{10}$ since only $\phi_n$ is needed to compute $C_{2n}$. On the other hand, computing $C_{11}$ using \eqref{eq:expC_11} requires computing $\phi_6$ and $\phi_7$, and it is therefore preferable to use Wigner's $(2n+1)$ rule that allows computing $C_{11}$ from $\phi_3$, $\phi_4$ and $\phi_5$.

\medskip

\noindent
{\bf Computation of higher-order terms.} For $n \ge 6$, the right-hand side of~\eqref{eqn:PsiNEqnAnyOrderH0} contains terms of the form $\mathcal{B}^{(k)} \phi_{n-k}$ with $k \ge 3$ and $n-k \ge 1$. The computation of $\phi_n$ therefore requires solving 2D boundary value problems of the form
$$
- \frac 12 \Delta T + \left( \kappa_{l_1}(r_1)+\kappa_{l_2}(r_2) \right) T = r_1^{l_1'} r_2^{l_2'} T^{(n-k)}_{(l_1'',l_2'')}
$$
for some $(l_1,l_2) \in {\cal L}_n$, $l_1'+l_2'=k-1$ and $(l_1'',l_2'') \in {\cal L}_{n-k}$. The right-hand side of this equation is not explicit, but the above equation can nevertheless be solved numerically since $T^{(n-k)}_{(l_1'',l_2'')}$ has been previously computed numerically during the calculation of $\phi_{n-k}$.

\section{Numerical results}
\label{sec:Cn}

\subsection{Comparison between different approaches}
\label{sect:comparisonOfResults}

The following tables contain the results of the approximated values of the $C_n$
coefficients computed by Ovsiannikov and Mitroy~\cite{ovsiannikov2005regular},
by Choy~\cite{Choy2000}, by Pauling and Beach~\cite{pauling1935van},
and by the techniques described in this paper.
The latter consist in solving recursively the Modified Slater--Kirkwood
boundary value problems of type~\eqref{eqn:PsiNEqnAnyOrderH0} using a Galerkin scheme in finite-dimensional approximation spaces constructed from tensor products of 1D Laguerre functions with degrees lower of equal to $k$. With basic double-precision floating-point arithmetics, the latter approach is numerical stable up to $k=11$ and provides results with excellent precision (relative error lower than $10^{-9}$). It is well-known that the conditioning of spectral methods for PDEs using orthogonal polynomial spaces grows exponentially. However, in the present case, the entries of the Galerkin matrix are square roots of rational numbers so that arbitrary precision can be obtained using symbolic computation. The method of Choy~\cite{Choy2000} is based on the Slater--Kirkwood
algorithm~\cite{ref:SlaKervdWgas}, whereas the method of Pauling and
Beach~\cite{pauling1935van} is different.
Although Slater and Kirkwood are referenced in~\cite{pauling1935van}, Pauling and Beach
were motivated by a method of S.~C.~Wang~\cite{wangunkowname}.

\begin{table}[H]
\begin{center}
\begin{tabular}{|c|l|l|l|l|}
\hline
Method & $C_{{6}}$          & $C_{{8}}$              & $C_{{10}}$           & $C_{{11}}$ \\
\hline
{\cite{pauling1935van}}   & 6.49903          & 124.399        & 1135.21       & \\
{\cite{Choy2000}} & 6.4990267        & 124.3990835    & 1135.2140398       & \\
\hline
{This work} & 6.49902670540 {\hspace{2mm}\cite{lrsBIBhi}} & 124.399083 & 3285.82841 & -3474.89803 \\
{\cite{ovsiannikov2005regular}} & 6.499026705406    & 124.3990835836   & 3285.828414967     & -3474.898037882 \\
    \hline
  \end{tabular}
\end{center}
\caption{Comparison of the coefficients $C_6$ to $C_{11}$ between various
papers and the basis states method and our method based on numerical solutions of boundary value problems of type~\eqref{eqn:TMain} in tensor products of Laguerre functions up to degree 11 (for which round-off error is suitably controlled). These results agree at  least to 9 digits with the results in~\cite{cebim2009high,mitroy2005higher,ovsiannikov2005regular,thakkar1988higher,yan1999third}.}
\end{table}

The discrepancy between the
Choy and Pauling--Beach results (who agree to the digits given)
and the other methods for $C_{10}$ has the following origin.  According to \eqref{eq:expC_10},
we have
$$
C_{10}= 2 \beta^{(5)}_{(1,3)} t^{(5)}_{(1,3)} +  \beta^{(5)}_{(2,2)} t^{(5)}_{(2,2)}.
$$
It appears that Choy in~\cite{Choy2000}, who also was guided by~\cite{ref:SlaKervdWgas},  only computed the second term
\begin{equation}
 \beta^{(5)}_{(2,2)} t^{(5)}_{(2,2)} =1135.214\dots
\end{equation}

\begin{table}[H]
\begin{center}
  \begin{tabular}{|c|l|l|l|l|}
    \hline
    Method            & $C_{{12}}$           & $C_{{13}}$            & $C_{{14}}$   &     $C_{{15}}$    \\
    \hline
{This work} & 122727.608 & -326986.924 & 6361736.04 &  -28395580.6 \\
   {\cite{ovsiannikov2005regular}}  & 122727.6087007     & -326986.9240441     & 6361736.045092  & -28395580.6   \\
    \hline
  \end{tabular}
\end{center}
\caption{Comparison of the $C_n$ coefficients $C_{12}$ to $C_{15}$  between~\cite{ovsiannikov2005regular} and our method based on numerical solutions of boundary value problems of type~\eqref{eqn:TMain} in tensor products of Laguerre functions up to degree 11 (for which round-off error is suitably controlled). These results agree at  least to 9 digits with the results in~\cite{mitroy2005higher,ovsiannikov2005regular,yan1999third} for $C_{13}$ and $C_{15}$ and~\cite{mitroy2005higher,ovsiannikov2005regular} for $C_{12}$ and $C_{14}$.}
\end{table}

\begin{table}[H]
\begin{center}
  \begin{tabular}{|c|l|l|l|l|}
    \hline
    Method            & $C_{{16}}$           & $C_{{17}} \times 10^{-9}$            & $C_{{18}} \times 10^{-10} $   &     $C_{{19}} \times 10^{-11}$    \\
    \hline
    {This work} & 441205192 & -2.73928165 & 3.93524773 & -3.07082459 \\
    {\cite{ovsiannikov2005regular}}  & 441205192.2739     & -2.739281653140  & 3.93524773346 & -3.07082459389 \\
    \hline
  \end{tabular}
\end{center}
\caption{Comparison of the $C_n$ coefficients $C_{16}$ to $C_{19}$  between~\cite{ovsiannikov2005regular} and our method based on numerical solutions of boundary value problems of type~\eqref{eqn:TMain} in tensor products of Laguerre functions up to degree 11 (for which round-off error is suitably controlled). These results agree at  least to 9 digits with the results in~\cite{mitroy2005higher,ovsiannikov2005regular}.}
\end{table}

\subsection{Role of continuous spectra in sum-over-state formulae}

It follows from \eqref{eq:eqphi3}, \eqref{eq:normphi3} and \eqref{eq:expC_6} that the leading coefficient $C_6$ of the van der Waals expansion can be written as
$$
C_6 = \langle \Bthree\phi_0, (H_0-\lambda_0)_{\phi_0^\perp}^{-1} \Bthree\phi_0 \rangle,
$$
where $(H_0-\lambda_0)_{\phi_0^\perp}^{-1}$ is the inverse of the restriction to $H_0-\lambda_0$ to the invariant subspace $\phi_0^\perp$ (which is well-defined since $\lambda_0$ is a non-degenerate eigenvalue of the self-adjoint operator $H_0$. This expression is sometimes wrongly rewritten as a sum-over-state formula
\begin{equation}\label{eq:C6sos}
C_6 = \sum_{j} \frac{|\langle \psi_j, \Bthree\psi_0\rangle|^2}{E_j-E_0} \quad \mbox{(wrong)},
\end{equation}
with $\psi_0:=\phi_0$, $E_0:=\lambda_0=-1$, where the $\psi_j$'s form an orthonormal family of excited states of $H_0$ associated with the eigenvalues $E_j$. This is not possible because $H_0$ has a non-empty continuous spectrum. Using \eqref{eq:C6sos} with a sum running over the excited states of $H_0$ (and omitting an integral over the scattering states of $H_0$) leads to an error that we are going to estimate.
We have
$$
C_6' := \sum_{j} \frac{|\langle \psi_j, \Bthree\psi_0\rangle|^2}{E_j-E_0} = - \langle \Bthree \phi_0, \phi_{3,\rm pp} \rangle,
$$
where $\phi_{3,\rm pp}$ is the projection of $\phi_3$ on the Hilbert space spanned by the eigenfunctions of $H_0$. Recall that the eigenvalues and associated eigenfunctions of the  hydrogen atom Hamiltonian $h_0:=- \frac{1}{2}\Delta - \frac{1}{|\mathbf{r}|}$, which is a self-adjoint operator on $L^2(\R^3)$, are of the form
\begin{equation}
\varepsilon_n = -\frac{1}{2n^2}, \quad  \psi_{n,l,m}(\mathbf{r}) = \varphi_{n,l}(r){Y}_l^m(\theta, \phi), \quad n \in \N^\ast, \quad 0 \le l \le n - 1, \quad -l \le m \le l,
\end{equation}
with
\begin{equation}\label{eqn:singleHydrogenAtomRadialPart}
 \varphi_{n,1} =  \sqrt{\left(\frac{2}{n}\right)^3 \frac{(n-2)!}{2n(n+1)!}}
  \left(\frac{2r}{n}\right) L_{n-2}^{(3)}\left(\frac{2r}{n}\right) e^{-r/n},
\end{equation}
where the associated Laguerre polynomials of the second type $L_n^{(m)}$, $n,m \in \N$, are defined from the Laguerre polynomial $L_n$ and are given by
\begin{equation}
  L^{(m)}_n (x) = (-1)^{m} \frac{\dd^m L_{n+m}}{\dd x^m}(x)
  = \frac 1 {n!} \sum_{k=0}^n \frac{n!}{k!} \begin{pmatrix} n+ m \\ n-k \end{pmatrix} (-x)^k.
\end{equation}
The eigenvalues and associated eigenfunctions of $H_0$ are therefore given by
$$
{\cal E}_{n_1,n_2}=\varepsilon_{n_1}+ \varepsilon_{n_2}= -\frac{1}{2n_1^2}-\frac{1}{2n_2^2}, \qquad \Psi_{n_1,l_1,m_1;n_2,l_2,m_2} = \psi_{n_1,l_1,m_1} \otimes \psi_{n_2,l_2,m_2},
$$
for $n_j \in \N^\ast$, $0 \le l_j \le n_j - 1$,  $-l_j \le m_j \le l_j$. Note that $\phi_0=\Psi_{1,0,0;1,0,0}$.
We therefore have
\begin{align*}
C_6' &=  \sum_{(n_1,n_2) \in (\N^* \times \N^*)\setminus \{(1,1)\}} \sum_{l_1=0}^{n_1-1} \sum_{l_2=0}^{n_2-1} \sum_{m_1=-l_1}^{l_1}
 \sum_{m_2=-l_2}^{l_2}  \frac{|\langle \Psi_{n_1,l_1,m_1;n_2,l_2,m_2}, \Bthree\psi_0\rangle|^2}{\varepsilon_{n_1}+ \varepsilon_{n_2}+1},
 \end{align*}
Using~\eqref{eq:B3} and the $L^2({\mathbb S}^2)$-orthonormality of the spherical harmonics, we get
$$
\langle \Psi_{n_1,l_1,m_1;n_2,l_2,m_2}, \Bthree\psi_0\rangle = \pi^{-1} S_{n_1} S_{n_2} \sum_{m=-1}^1 G_{\rm c}(1,1,m) \delta_{l_1,1} \delta_{l_2,1} \delta_{m,m_1} \delta_{{-m},m_2},
$$
where
\begin{equation}\label{eq:Sn}
S_n:=\int_0^{+\infty} r^3 e^{-r} \phi_{n,1}(r) \, \dd r = 8 n^3 \frac{(n-1)^{n-3}}{(n+1)^{n + 3}}   \sqrt{\frac{(n+1)!}{(n-2)!}} .
\end{equation}
The latter expression is derived in Appendix~\ref{sec:Sn}. We finally obtain
\begin{equation}\label{eq:tildeC6}
 C_6'  = \pi^{-2} \sum_{m=-1}^1 |G_{\rm c}(1,1,m)|^2 \sum_{n_1,n_2 \ge 2} \frac{S_{n_1}^2S_{n_2}^2}{1- \frac{1}{2n_1^2}- \frac{1}{2n_2^2}} = \frac{32}3 \sum_{n_1,n_2 \ge 2} \frac{S_{n_1}^2S_{n_2}^2}{1- \frac{1}{2n_1^2}- \frac{1}{2n_2^2}} .
\end{equation}
Summing up the terms of the above series for $n_1,n_2 \le 300$ (note that $S_n \sim_{n \to \infty} \frac{8}{e^2n^{3/2}}$), we obtain the approximate value
$$
 C_6'  \simeq   3.923  
$$
which shows that the continuous spectrum plays a major role in the sum-over-state evaluation of the $C_6$ coefficient of the hydrogen molecule (recall that $C_6 \simeq 6.499$).

\section{Proofs}
\label{sec:proofs}

We now establish the results stated above, starting from Lemma~\ref{lem:f_to_T}.

\subsection{Proof of Lemma~\ref{lem:f_to_T}}
\label{sec:proof_lemRS}

Recall that the Hydrogen atom Hamiltonian $h_0=- \frac{1}{2}\Delta - \frac{1}{|\mathbf{r}|}$  introduced in the previous section is a self-adjoint operator on $L^2(\R^3)$ with domain $H^2(\R^3)$, and that its ground state is non-degenerate:
$$
h_0\psi_{1,0,0} = - \frac 12 \psi_{1,0,0} \quad \mbox{with} \quad \psi_{1,0,0} = \varphi_{1,0}(r) {Y}_0^0(\theta,\phi) = \pi^{-1/2} e^{-r}, \quad \|\psi_{1,0,0}\|_{L^2(\R^3)}=1.
$$
Since $H_0=h_0 \otimes \1_{L^2(\R^3)} + \1_{L^2(\R^3)}  \otimes h_0$, $H_0$ is a self-adjoint operator on $L^2(\R^3 \times \R^3)$ with domain $H^2(\R^3 \times \R^3)$ and it also has a non-degenerate ground state
$$
H_0\phi_0= \lambda_0\phi_0 \quad \mbox{with} \quad \phi_0=\psi_{1,0,0} \otimes \psi_{1,0,0}=\pi^{-1} e^{-(r_1+r_2)}, \quad \|\phi_0\|_{L^2(\R^3\times\R^3)}=1 \quad \mbox{and} \quad \lambda_0=-1.
$$
Given $(\alpha,F) \in \R \times L^2(\R^3 \times \R^3)$, the problem consisting of seeking $(\mu,\Psi) \in \R \times H^2(\R^3 \times \R^3)$ such that
\begin{equation}\label{eq:PRS}
(H_0-\lambda_0) \Psi = F-\mu \phi_0, \quad \langle \phi_0,\Psi \rangle = \alpha,
\end{equation}
is well-posed and its unique solution is given
$$
\Psi = (H_0-\lambda_0)|_{\phi_0^\perp}^{-1} \Pi_{\phi_0^\perp} F + \alpha \phi_0, \quad \mu= \langle \phi_0,F \rangle,
$$
where $(H_0-\lambda_0)|_{\phi_0^\perp}^{-1}$ is the inverse of $H_0-\lambda_0$ on the invariant subspace $\phi_0^\perp$ and $\Pi_{\phi_0^\perp} F := F-\langle \phi_0,F \rangle\phi_0$ the orthogonal projection of $F$ on  $\phi_0^\perp$. Consider the unitary map
$$
{\cal U} : L^2(\Omega) \otimes L^2({\mathbb S}^2) \otimes L^2({\mathbb S}^2) \to L^2(\R^3 \times \R^3) \equiv L^2(\R^3) \otimes L^2(\R^3)
$$
induced by the spherical coordinates defined for all $f \in L^2(\Omega)$, $l_1,l_2 \in \N$, $-l_j \le m_j \le l_j$ by
$$
({\cal U}(f \otimes s_1 \otimes s_2))(\br_1,\br_2) = \frac{f(|\br_1|,|\br_2|)}{|\br_1| \, |\br_2|} \, s_1\left( \frac{\br_1}{|\br_1|} \right) \, s_2\left( \frac{\br_2}{|\br_2|} \right).
$$
Since $({Y}_l^m)_{l \in \N, \, -l \le m \le l}$ is an orthonormal basis of $L^2({\mathbb S}^2)$, we have
$$
L^2(\Omega) \otimes L^2({\mathbb S}^2) \otimes L^2({\mathbb S}^2) = \bigoplus_{l_1,l_2 \in \N} \bigoplus_{m_1=-l_1}^{l_1} \bigoplus_{m_2=-l_2}^{l_2}  {\cal H}_{l_1,l_2}^{m_1,m_2}
$$
where ${\cal H}_{l_1,l_2}^{m_1,m_2} :=L^2(\Omega) \otimes \C Y_{l_1}^{m_1} \otimes \C Y_{l_2}^{m_2}$. It follows from classical results for Schr\"odinger operators on $L^2(\R^3)$ with central potentials (see e.g.~\cite[Section XIII.3.B]{reed1979methods}) that each ${\cal H}_{l_1,l_2}^{m_1,m_2}$ is an invariant subspace for ${\cal U}^* H_0 {\cal U}$ and that
$$
{\cal U}^* H_0 {\cal U}|_{{\cal H}_{l_1,l_2}^{m_1,m_2}}
=  H_{l_1,l_2} \otimes  \1_{\C { Y}_{l_1}^{m_1}} \otimes  \1_{\C { Y}_{l_2}^{m_2}},
$$
where the expression of $H_{l_1,l_2}$ can be derived by adapted the arguments in~\cite[Section 3]{lrsBIBhi}, that we do not detail here for the sake of brevity:
$H_{l_1,l_2}$ is the self-adjoint operator on $L^2(\Omega)$ with form domain $H_1^0(\Omega)$ defined by
\begin{equation}\label{eq:Hl1l2}
H_{l_1,l_2}  = - \frac 12 \Delta + \kappa_{l_1}(r_1)+\kappa_{l_2}(r_2) + \lambda_0.
\end{equation}
Note that the operator $H_{l_1,l_2}$ on $L^2(\Omega) \equiv L^2(0,+\infty) \otimes L^2(0,+\infty)$ can itself be decomposed as
$$
H_{l_1,l_2}=h_{l_1} \otimes \1_{L^2(0,+\infty)} +\1_{L^2(0,+\infty)} \otimes h_{l_2} \ge - \frac 1{2(l_1+1)^2}- \frac 1{2(l_2+1)^2},
$$
where for each $l \in \N$, $h_{l}$ is the self-adjoint operator on $L^2(0,+\infty)$ with form domain $H^1_0(0,+\infty)$ defined by
$$
h_l := - \frac 12 \frac{d^2}{dr^2} + \frac{l ( l+1)}{2r^2} - \frac{1}{r} = - \frac 12 \frac{d^2}{dr^2}+\kappa_l - \frac 12.
$$
This well-known operator allows one to construct the bound-states of hydrogen atom with orbital quantum number $l$. It satisfies
$h_l \ge - \frac 1{2(l+1)^2}$ and its ground state eigenvalue $- \frac 1{2(l+1)^2}$ is non-degenerate. It follows from this bound that
\begin{equation}\label{eq:boundHl1l22}
H_{l_1,l_2}-\lambda_0 = H_{l_1,l_2}+1 \ge \frac 38 \quad \mbox{for all } (l_1,l_2) \in \mathbb{N}^2 \setminus \{(0,0\}.
\end{equation}
Choosing $\alpha=0$ in \eqref{eq:PRS} amounts to enforcing that the solution $\Psi$ is in $\phi_0^\perp$. Taking $\alpha=0$ and $F= \frac{f(r_1,r_2)}{r_1r_2} {Y}_{l_1}^{m_1}(\theta_1, \phi_1) {Y}_{l_2}^{m_2}(\theta_2, \phi_2) = {\cal U}(f \otimes { Y}_{l_1}^{m_1} \otimes { Y}_{l_2}^{m_2})$, with $f \in L^2(\Omega)$, it follows that \eqref{eq:eq_res} has a unique solution in $H^2(\R^3 \times \R^3)$ if and only if $\mu=\langle \phi_0, F\rangle=0$, that is
$$
\delta_{(l_1,l_2)=(0,0)} \int_\Omega f(r_1,r_2) e^{-(r_1+r_2)} r_1 r_2 \, dr_1 \, dr_2 =0,
$$
in which case the solution is given by $\Psi={\cal U}(T \otimes { Y}_{l_1}^{m_1} \otimes { Y}_{l_2}^{m_2})$ where
\begin{align*}
&T:= (H_{l_1,l_2}-\lambda_0)^{-1} f  & \mbox{if } (l_1,l_2) \neq (0,0), \\
&T:=(H_{0,0}-\lambda_0)|_{(r_1r_2e^{-(r_1+r_2)})^\perp}^{-1} f  & \mbox{if } (l_1,l_2) = (0,0).
\end{align*}
We therefore have
$$
\psi =   \frac{T(r_1,r_2)}{r_1r_2} {Y}_{l_1}^{m_1}(\theta_1, \phi_1) {Y}_{l_2}^{m_2}(\theta_2, \phi_2),
$$
where $T$ is the unique solution to \eqref{eqn:TMain} in $H^1_0(\Omega)$ if $(l_1,l_2) \neq (0,0)$ and $T$ is the unique solution to \eqref{eqn:TMain} in $\widetilde H^1_0(\Omega)=H^1_0(\Omega) \cap (r_1r_2e^{-(r_1+r_2)})^\perp$ if $(l_1,l_2)=0$.

\medskip

The fact that if $f$ decays exponentially at infinity, then so does $T$, hence $\psi$, is a consequence of the following result, whose proof follows the same lines as in~\cite[Section 3.3]{lrsBIBhi} where this result is established for the special case when $(l_1,l_2)=(1,1)$ and $f=r_1^2r_2^2e^{-(r_1+r_2)}$.

\begin{lemma}\label{lem:exponentialdecay}
If the function $f$ of \eqref{eqn:TMain} decays exponentially at infinity at a rate $\eta > 0$, in the sense that
\begin{equation}\label{eq:hypf}
 \Vert  e^{\eta ( r_1 + r_2)} f \Vert_{L^2 ( \Omega)} < \infty,
\end{equation}
then the unique solution $T$ of \eqref{eqn:TMain} also decays exponentially at infinity.
More precisely, for all $0 \le \alpha < \sqrt{3/8}$,  there exists a constant $C_\alpha \in \R_+$ such that for all $\eta > \alpha$ and all $f \in L^2(\Omega)$ satisfying \eqref{eq:hypf}, it holds
\begin{equation}
\Vert  e^{\alpha (r_1+ r_2)} T \Vert_{H^1 ( \Omega)} \le C_\alpha \Vert  e^{\eta ( r_1 + r_2)} f \Vert_{L^2 ( \Omega)}.
\end{equation}
\end{lemma}

\begin{proof} We limit ourselves to the case when $(l_1,l_2) \neq (0,0)$. The special case $(l_1,l_2) = (0,0)$ can be dealt with similarly, by replacing $H^1_0(\Omega)$ by $\widetilde H^1_0(\Omega)$.
Let $a$ be the continuous bilinear form on $H^1_0(\Omega) \times H^1_0(\Omega)$ associated with the positive self-adjoint operator $H_{l_1,l_2}-\lambda_0$:
$$
\forall u,v \in H^1_0(\Omega), \quad a(u,v) = \frac 12  \int_\Omega \nabla u \cdot \nabla v + \int_\Omega (\kappa_{l_1}(r_1)+\kappa_{l_2}(r_2)) u(r_1,r_2) v(r_1,r_2) \, dr_1 \, dr_2.
$$
Recall that the continuity of $a$ can be shown directly (without using the fact that $H^1_0(\Omega)$ is the form domain of $H_{l_1l_2}$) as a straightforward consequence of the one-dimensional Hardy inequality
\begin{equation} \label{eq:Hardy}
\forall g \in H^1_0(0,+\infty), \quad \int_0^\infty (g(r)/r)^2 \dd r \leq 4 \int_0^\infty g' ( r) ^2 \, \dd r.
\end{equation}
It follows from \eqref{eq:boundHl1l22} that $a \ge \frac 38$ (in the sense of quadratic forms on $L^2(\Omega)$).
For $0 \le \alpha < \sqrt{3/8}$, we introduce the continuous bilinear form $a_\alpha$ on $H^1_0(\Omega) \times H^1_0(\Omega)$ defined by
$$
\forall u,v \in H^1_0(\Omega), \quad   a_\alpha (u,v) = a(u,v)
  	   - \int_\Omega \alpha u( \mathbf{r}) \left( \frac{\partial v}{\partial r_1} ( \mathbf{r} ) + \frac{\partial v}{\partial r_2} ( \mathbf{r} ) \right) \dd \mathbf{r}
	    - \int_\Omega \alpha^2 u (\mathbf{r}) v (\mathbf{r}) \dd \mathbf{r},
 $$
 for which
 $$
 \forall v \in H^1_0(\Omega), \quad a_\alpha ( v, v) = a( v,v) -  \alpha^2 \|v\|_{L^2(\Omega)}^2
 \geq  \underbrace{\left( \frac 38- \alpha^2 \right)}_{>0} \|v\|_{L^2(\Omega)}^2.
$$
Using either the fact that $\kappa_l(r) \ge \frac 14$ (for $l \ge 1$) or the Hardy inequality~\eqref{eq:Hardy} (for $l=0$), we also have
$$
\forall v \in H^1_0(\Omega), \quad a_\alpha (v, v) = a( v,v) -  \alpha^2 \|v\|_{L^2(\Omega)}^2 \ge \frac 14 \int_{\Omega} |\nabla v|^2 - 2 \|v\|_{L^2}^2.
$$
Since $a \ge \frac 38$ and $a_\alpha \ge \left( \frac 38- \alpha^2 \right) > 0$, the above bound implies that $a$ and $a_\alpha$ are both continuous and coercive on $H^1_0(\Omega)$. The function $T \in H^1_0(\Omega)$ solution to \eqref{eqn:TMain} is also the unique solution to the variational equation
$$
\forall w \in H^1_0(\Omega), \quad a(T,w) = \int_\Omega fw.
$$
Proceeding as in~\cite[Section 3.3]{lrsBIBhi}, we obtain that for all $u \in H^1_0(\Omega)$ such that $e^{\alpha (r_1+r_2)}u \in H^1_0(\Omega)$ and $w \in C^\infty_{\rm c}(\Omega)$, we have
\begin{equation}\label{eq:a_alpha}
a_\alpha(e^{\alpha (r_1+r_2)}u,w)=a(u, e^{\alpha (r_1+r_2)}w).
\end{equation}
Consider now $f \in L^2(\Omega)$ satisfying \eqref{eq:hypf} for some $\eta >  \alpha$. The function
$e^{\alpha(r_1+r_2)}f$ is in $L^2(\Omega)$, so that the problem of finding $v \in H^1(\Omega)$ such that
$$
\forall w \in H^1_0(\Omega), \quad a_\alpha(v,w)=\int_{\Omega} e^{\alpha(r_1+r_2)}f w
$$
has a unique solution $v$, satisfying $\|v\|_{H^1(\Omega)} \le C_\alpha \|e^{\alpha(r_1+r_2)}f\|_{L^2(\Omega)} \le C_\alpha \|e^{\eta(r_1+r_2)}f\|_{L^2(\Omega)}$, where $C_\alpha \ge 1$ is the ratio between the continuity constant and the coercivity constant of $a_\alpha$. Let $u=e^{-\alpha(r_1+r_2)}v \in H^1_0(\Omega)$. In view of \eqref{eq:a_alpha}, we have
$$
\forall w \in C^\infty_{\rm c}(\Omega), \quad a(u,e^{\alpha(r_1+r_2)}w)=a_\alpha(v,w)=
\int_{\Omega} e^{\alpha(r_1+r_2)}f w = a(T,e^{\alpha(r_1+r_2)}w).
$$
Hence, $T=u$ and $\|e^{\alpha(r_1+r_2)} T\|_{H^1(\Omega)} = \|e^{\alpha(r_1+r_2)} u \|_{H^1(\Omega)} = \|v\|_{H^1(\Omega)} \le C_\alpha \|e^{\eta(r_1+r_2)}f\|_{L^2(\Omega)}$.
\end{proof}

\medskip

\noindent
As a consequence, we have
\begin{align*}
\|e^{\alpha(|\br_1|+|\br_2|)} \psi\|_{L^2(\Rsix)} &= \|e^{\alpha(r_1+r_2)} T\|_{L^2(\Omega)} \le \|e^{\alpha(r_1+r_2)} T\|_{H^1(\Omega)}  \\
&\le C_\alpha \|e^{\eta(r_1+r_2)}f\|_{L^2(\Omega)} = C_\alpha \|e^{\eta(|\br_1|+|\br_2|)}F\|_{L^2(\R^6)},
\end{align*}
which proves \eqref{eq:boundH1exp}. In addition, a simple calculation using \eqref{eq:Hardy} shows that for all $g \in H^1_0(\Omega)$
\begin{align*}
\left\| \frac{g}{r_1r_2} \otimes Y_{l_1}^{m_1} \otimes Y_{l_2}^{m_2} \right\|_{H^1(\Rsix)}^2
&=  \|g\|_{H^1(\Omega)}^2 + l_1(l_1+1) \left\| \frac{g}{r_1} \right\|_{L^2(\Omega)}^2 + l_2(l_2+1) \left\| \frac{g}{r_2} \right\|_{L^2(\Omega)}^2 \\
&\le (1+4 l_1(l_1+1) +4 l_2(l_2+1) ) \|g\|_{H^1}^2,
\end{align*}
yielding
\begin{align*}
\|e^{\alpha(|\br_1|+|\br_2|)} \psi\|_{H^1(\Rsix)} & \le  (1+4 l_1(l_1+1) +4 l_2(l_2+1))^{1/2} \|e^{\alpha(r_1+r_2)}T\|_{H^1(\Omega)} \\
&\le C_\alpha
(1+4 l_1(l_1+1) +4 l_2(l_2+1))^{1/2} \|e^{\eta(|\br_1|+|\br_2|)}F\|_{L^2(\Omega)}.
\end{align*}
Lastly, since $H_{l_1,l_2}$ is a real operator in the sense that $\overline{H_{l_1,l_2}\phi}=H_{l_1,l_2}\overline{\phi}$ for all $\phi \in D(H_{l_1,l_2})$, it is obvious that $T$ is real-valued, whenever $f$ is.

\subsection{Proof of Lemma~\ref{lem:PsiNEqnAnyOrder} and Theorem~\ref{thm:phin}}
\label{sec:proof_lem1}

We have seen in the previous section that for each $(\alpha,F) \in \R \times L^2(\R^3 \times \R^3)$, \eqref{eq:PRS} has a unique solution $(\mu,\psi)$ in $\R \times H^2(\R^3 \times \R^3)$. For $n=1$, we have
$$
 (H_0 - \lambda_0) \phi_1=-C_1 \phi_0, \quad \biginner{\phi_0}{\phi_1}=0,
 $$
 and it is clear that $(C_1,\phi_1)=(0,0)$ is {\it a} solution, hence {\it the} solution, to this system. Likewise, for $n=2$, we have
 $$
 (H_0 - \lambda_0) \phi_2=-C_1\phi_1 - C_2 \phi_0 = -C_2\phi_2, \quad \biginner{\phi_0}{\phi_2}=-\frac 12 \biginner{\phi_1}{\phi_1}=0,
 $$
 so that $(C_2,\phi_2)=(0,0)$. To prove that the Rayleigh--Schr\"odinger triangular system \eqref{eqn:PsiNEqnAnyOrderH0}-\eqref{eqn:PsiNEqnAnyOrderPsi0} is well-posed and that $\phi_n$ is of the form~\eqref{eqn:psiiform},  we proceed by induction on $n$.
 It is proven in~\cite{lrsBIBhi} that for $n=3$,
$$
 \phi_3 =   \frac{T^{(3)}_{(1,1)}( r_1, r_2)}{r_1 r_2} \sum_{m=-1}^1 \alpha^{(3)}_{(1,1,m)} {Y}_{1}^{m} (\theta_1, \phi_1) {Y}_{1}^{-m}  ( \theta_2, \phi_2) ,
 $$
 with $\alpha^{(3)}_{(1,1,m)} =-\pi G_{\rm c}(1,1,m)$ and $\| T^{(3)}_{(1,1)}( r_1, r_2) e^{\eta^3_{1,1}(r_1+r_2)} \|_{H_1(\Omega)}=:C^3_{1,1} < \infty$. Let ${\cal L}_3=\{(1,1)\}$ and assume that for some $n \ge 3$ the following recursion hypotheses are satisfied (this is the case for $n=3$): for all $3 \le k \le n$,
\begin{equation}\label{eqn:psiiform3rec}
 \phi_k = \sum_{(l_1,l_2) \in {\cal L}_k}  \frac{T^{(k)}_{(l_1,l_2)}( r_1, r_2)}{r_1 r_2} \left(  \sum_{m=-\min(l_1,l_2)}^{\min(l_1,l_2)} \alpha^{(k)}_{(l_1,l_2,m)} {Y}_{l_1}^{m} (\theta_1, \phi_1) {Y}_{l_2}^{-m}  ( \theta_2, \phi_2) \right),
\end{equation}
for some finite set ${\cal L}_k \subset \N^2$ with cardinality $N_k<\infty$, where $T^{(k)}_{(l_1,l_2)}$ is the unique solution to \eqref{eqn:TMain} in $H^1(\Omega)$ (or in $\widetilde H^1(\Omega)$ if $l_1=l_2=0$) for $f=f^{(k)}_{(l_1,l_2)} \in L^2(\Omega)$ and that for all $(l_1,l_2) \in {\cal L}_k$ there exists $\eta^k_{l_1,l_2}>0$ such that
\begin{equation}\label{eq:decayTj}
\| T^{(k)}_{(l_1,l_2)}( r_1, r_2) e^{\eta^k_{l_1,l_2}(r_1+r_2)} \|_{H^1(\Omega)}=:C^k_{l_1,l_2} < \infty.
\end{equation}
From \eqref{eqnBic}, the fact that $\phi_1=\phi_2=0$ and the recursion hypothesis~\eqref{eqn:psiiform3rec}, we obtain that for all $3 \le k \le n+1$,
\begin{align}
\mathcal{B}^{(k)} \phi_{n+1-k} =   &\sum_{\substack{ l_1 + l_2 = k-1 \\ l_1, l_2 \neq 0 } } \sum_{(l_1',l_2') \in {\cal L}_{n+1-k}}
\sum_{m=-\min(l_1,l_2)}^{\min(l_1,l_2)}      \sum_{m'=-\min(l_1',l_2')}^{\min(l_1',l_2')}   \nonumber \\
& \qquad {\cal U} \left( f_{n-k+1,l_1,l_1',l_2,l_2'}^{m,m'}  \otimes  {Y}_{l_1}^m {Y}_{l_1'}^{m'} \otimes
{Y}_{l_2}^{-m} {Y}_{l_2'}^{-m'} \right), \label{eq:Bphi}
\end{align}
where
$$
 f_{j,l_1,l_1',l_2,l_2'}^{m,m'}(r_1,r_2)  :=  G_{\rm c}(l_1,l_2, m)  r_1^{l_1}r_2^{l_2}
 \alpha^{(j)}_{(l_1',l_2',m')} T^{(j)}_{(l_1',l_2')}(r_1, r_2).
$$
In addition, we have
\begin{equation}\label{eq:prod_Ylm}
 {Y}_{l}^m {Y}_{l'}^{m'} = \sum_{l''=|l-l'|}^{l+l'} \zeta_{l,l',l''}^{m,m'} {Y}_{l''}^{m+m'} \quad \mbox{where} \quad  \zeta_{l,l',l''}=0 \mbox{ if } l+l'+l'' \notin 2\N,
\end{equation}
where the coefficients $\zeta_{l,l',l''}^{m,m'} \in \R$ can be computed explicitly using Wigner's 3-j symbols~\cite[p. 146]{brinki968}:
\[
\zeta_{l,l',l''}^{m, m'} = (-1)^{m + m'}
\sqrt{\frac{(2l+1) (2l' + 1) (2l'' +1)}{4\pi}}
\begin{pmatrix}
  l & l' & l'' \\
  0   & 0   & 0
\end{pmatrix}
\begin{pmatrix}
  l & l' & l'' \\
  m & m' & -m-m'
\end{pmatrix}.
\]
This implies that
\begin{align}
&- \sum_{k = 3}^{n+1} \mathcal{B}^{(k)} \phi_{n+1-k},
      - \sum_{k = 1}^{n+1} C_k \phi_{n+1-k}  \nonumber \\
      &\qquad \qquad =
       \sum_{(l_1,l_2) \in {\cal L}_{n+1}}  \frac{f^{(n+1)}_{(l_1,l_2)}( r_1, r_2)}{r_1 r_2} \left(  \sum_{m=-\min(l_1,l_2)}^{\min(l_1,l_2)} \alpha^{(k)}_{(l_1,l_2,m)} {Y}_{l_1}^{m} (\theta_1, \phi_1) {Y}_{l_2}^{-m}  ( \theta_2, \phi_2) \right), \label{eq:RHS_RS}
\end{align}
for some ${\cal L}_{n+1} \subset \N^2$ with finite cardinality, where the $f^{(n+1)}_{l_1,l_2}$'s are linear combinations of the functions $r_1^{l_1}r_2^{l_2}T^{(j)}_{(l_1',l_2')} \in L^2(\Omega)$, $3 \le j \le n$, $l_1',l_2' \in {\cal L}_j$, $l_1+l_2+j \le n+1$, and therefore satisfy in view of \eqref{eq:decayTj}
\begin{equation}\label{eq:decayfj}
\| f^{(n+1)}_{(l_1,l_2)}( r_1, r_2) e^{\xi^{n+1}_{l_1,l_2}(r_1+r_2)} \|_{H_1(\Omega)} < \infty
\end{equation}
for some $\xi^{n+1}_{l_1,l_2} > 0$. Therefore the problem consisting in seeking $(C_{n+1},\phi_{n+1}) \in \R \times H^2(\R^3 \times \R^3)$ satisfying
$$
(H_0-\lambda_0)\phi_{n+1}=- \sum_{k = 3}^{n+1} \mathcal{B}^{(k)} \phi_{n+1-k},
      - \sum_{k = 1}^{n+1} C_k \phi_{n+1-k}, \qquad  \biginner{\phi_0}{\phi_{n+1}} = - \frac 12 \sum_{k = 1}^{n}
      \biginner{\phi_k}{\phi_{n+1-k}}
$$
is well-posed and we deduce from Lemma~\ref{lem:f_to_T} that
$$
\phi_{n+1} := \sum_{(l_1,l_2) \in {\cal L}_{n+1}}  \frac{T^{(n+1)}_{(l_1,l_2)}( r_1, r_2)}{r_1 r_2} \left(  \sum_{m=-\min(l_1,l_2)}^{\min(l_1,l_2)} \alpha^{(k)}_{(l_1,l_2,m)} {Y}_{l_1}^{m} (\theta_1, \phi_1) {Y}_{l_2}^{-m}  ( \theta_2, \phi_2) \right),
$$
where $T^{(n+1)}_{(l_1,l_2)}$ is the unique solution to \eqref{eqn:TMain} in $H^1(\Omega)$ (or in $\widetilde H^1(\Omega)$ if $l_1=l_2=0$) for $f=f^{(n+1)}_{(l_1,l_2)}$. In addition, it follows from \eqref{eq:decayfj} that~\eqref{eq:decayTj} holds true for $k=n+1$. Therefore, the Rayleigh--Schr\"odinger triangular system \eqref{eqn:PsiNEqnAnyOrderH0}-\eqref{eqn:PsiNEqnAnyOrderPsi0} is well-posed and the $T^{(n)}_{(l_1,l_2)}$'s decay exponentially at infinity in the sense of~\eqref{eq:decayTj}. From~\eqref{eqn:psiiform3rec} we obtain that for $\alpha_n=\min_{(l_1,l_2) \in {\cal L}_n}(\eta^n_{l_1,l_2})>0$, we have
\begin{align*}
\|e^{\alpha_n (|\br_1|+|\br_2|)} \phi_n \|_{H^1(\Rsix)} &\le C_n \sum_{(l_1,l_2) \in {\cal L}_n}  \| e^{\alpha_n (r_1+r_2)} T^{(n)}_{(l_1,l_2)} \|_{H^1(\Omega)} \\
&\le C_n \sum_{(l_1,l_2) \in {\cal L}_n}  \| e^{\eta^n_{(l_1,l_2)} (r_1+r_2)} T^{(n)}_{(l_1,l_2)}\|_{H^1(\Omega)} < \infty,
\end{align*}
for some $C_n \in \R_+$, so that $\phi_n$ decays exponentially at infinity in the sense of~\eqref{eq:exp_dec_phin}. Lastly, we infer from Wigner's $(2n+1)$ rule and the fact that $\phi_1=\phi_2=0$, that $C_n=0$ for $1 \le n \le 5$. This completes the proof of both Lemma~\ref{lem:PsiNEqnAnyOrder} and Theorem~\ref{thm:phin}.

\medskip

\noindent
Let us finally explain how to construct Table~\ref{tbl:lsphericalharmonicspsii}. We have already shown that ${\cal L}_3=\{(1,1)\}$, and from \eqref{eq:Bphi}-\eqref{eq:RHS_RS} and the fact that $\phi_1=\phi_2=0$, we see that
$$
{\cal L}_{n+1} \subset \left( \bigcup_{k=3}^{n-2} {\cal M}_{k,n+1-k} \right) \bigcup {\cal M}_{n+1,0} \bigcup \left( \bigcup_{3 \le k \le n-5 \; | \; C_{n+1-k} \neq 0} {\cal L}_k \right),
$$
where for $k,n \ge 3$,
\begin{align*}
{\cal M}_{k,0}&=\left\{ (l_1,l_2) \in \N^* \times \N^\ast \; | \; l_1+l_2=k-1 \right\} = \{(1,k-2),\cdots, (k-2,1)\}, \\
{\cal M}_{k,n} &= \big\{ (l_1,l_2) \in \N \times \N  \; | \; \exists (l_1',l_2') \in {\cal M}_{k,0}, \; \exists (l_1'',l_2'') \in {\cal L}_{n} \mbox{ s.t. } \\ & \qquad\qquad
|l_j'-l_j''| \le l_j \le l_j'+l_j'', \; l_j+l_j'+l_j'' \in 2 \N, \; j=1,2 \big\}.
\end{align*}
Consequently, we have
\begin{align*}
{\cal L}_4 &= {\cal M}_{4,0}; \\
{\cal L}_5 &= {\cal M}_{5,0}; \\
{\cal L}_6 &= {\cal M}_{3,3} \cup {\cal M}_{6,0} \quad \mbox{with} \quad  {\cal M}_{3,3} =\{(0,2;0,2)\}; \\
{\cal L}_7 &= {\cal M}_{3,4} \cup {\cal M}_{4,3} \cup {\cal M}_{7,0}  \quad \mbox{with} \quad   {\cal M}_{3,4} = {\cal M}_{4,3} = \{(0,2;1,3),(1,3;0,2)\}; \\
{\cal L}_8 &= {\cal M}_{3,5} \cup {\cal M}_{4,4} \cup {\cal M}_{5,3}  \cup {\cal M}_{8,0} \quad \mbox{with} \quad {\cal M}_{3,5} ={\cal M}_{5,3}= \{(0,2;2,4),(1,3;1,3),(2,4;0,2)\}, \\
&\qquad  {\cal M}_{4,4}=  \{(0,2;0,2,4), (0,2,4;0,2), (1,3;1,3)\}  \\
{\cal L}_9 &= {\cal M}_{3,6} \cup {\cal M}_{4,5} \cup {\cal M}_{5,4}  \cup {\cal M}_{6,3}  \cup {\cal M}_{9,0} \cup {\cal L}_3 \quad \mbox{with}  \\
&\qquad  {\cal M}_{6,3} \subsetneq {\cal M}_{3,6} =\{(0,2;3,5),(1,3;2,4),(2,4;1,3),(3,5;0,2), (1,3; 1,3)\}, \\
&\qquad {\cal M}_{4,5} ={\cal M}_{5,4}=\{(0,2;1,3,5),(1,3;0,2,4),(2,4;1,3),(1,3;2,4),(0,2,4;1,3),(1,3,5;0,2) \},
\end{align*}
where we recall that $(l_1,l_1';l_2,l_2')$ (resp. $(l_1,l_1';l_2,l_2',l_2'')$, $(l_1,l_1',l_1'';l_2,l_2')$) stands for the four (resp. six) pairs $(l_1,l_2)$, $(l_1',l_2)$, $(l_1,l_2')$, etc. After eliminating redundancies, we obtain Table~\ref{tbl:lsphericalharmonicspsii}.

\subsection{Proof of Theorem \ref{thm:treophmole} \label{sec:proof_treophmole}}

As in~\cite{lrsBIBhi}, we introduce the space
  \begin{equation} \label{eqn:symmspacerp}
  {\cal V}=\set{v\in L^2(\Rsix)}{ v(\bR_1,\bR_2)=v(\bR_2,\bR_1) \;\forall \bR_1,\bR_2\,\in\Rtre},
  \end{equation}
  the functions $\pe^{(n)} \in {\cal V} \cap H^2 (\Rsix)$ normalized in $L^2 ( \Rsix)$,
  \begin{equation}\label{eqn:normltestfn}
    \pe^{(n)} := \me^{(n)} \Te \left(\phi_\epsilon^{(n)}  \right) \quad \mbox{where} \quad
    \phi_\epsilon^{(n)}:=\phi_0 + \sum_{k = 3}^n \epsilon^{k} \phi_k
    \mbox{ and } \me^{(n)} =  \left\| \Te \left(\phi_\epsilon^{(n)} \right) \right\|_{L^2 ( \Rsix)}^{-1},
  \end{equation}
as well as the Rayleigh quotient
  \begin{equation}
    \mu_\epsilon^{(n)} = \inner{\pe^{(n)}}{H_\epsilon \pe^{(n)}}
  \end{equation}
  and the approximation
  $$
  \lambda^{(n)}_\epsilon = \lambda_0 - \sum_{k=6}^n C_n \epsilon^n
  $$
of $\lambda_\epsilon$.  When $\epsilon \to 0$, we have $\Te \left( \phi_0 \right) \to 1$ and therefore $m_\epsilon^{(n)} \to 1$. We know from~\cite[Section 2.4]{lrsBIBhi} that there exists a constant $C \in \R_+$ such that for $\epsilon > 0$ small enough
  $$
  \|\psi_\epsilon - \pe^{(3)}\|_{H^2(\Rsix)} \le C \epsilon^4, \quad |\lambda_\epsilon- \mu^{(3)}_\epsilon| \le C \epsilon^8, \quad \mbox{and} \quad |\lambda_\epsilon- \lambda^{(6)}_\epsilon| \le C \epsilon^7.
  $$
  It follows from Theorem~\ref{thm:phin} that the $\phi_n$'s are in $H^2(\Rsix)$. Since $\Te$ continuous on this space, we obtain that for all $n \ge 3$, there exists $c_n \in \R$, such that for $\epsilon > 0$ small enough
  $$
  \|\psi_\epsilon - \pe^{(n)}\|_{H^2(\Rsix)} \le c_n \epsilon^4.
  $$
  We infer from~\cite[Lemma 2.2 and Appendix A]{lrsBIBhi} that there exists a constant $C \in \R_+$ such that for all $n \ge 3$ there exists $\epsilon > 0$ such that for all  $0 < \epsilon \le \epsilon_n$,
    \begin{align} \label{eq:lambda_eps}
  |\lambda_\epsilon - \mu^{(n)}_\epsilon | &\le C \|H_\epsilon \psi^{(n)}_\epsilon - \mu^{(n)}_\epsilon\psi^{(n)}_\epsilon \|_{L^2(\Rsix)}^2, \\
  \|  \psi_\epsilon - \pe^{(n)} \|_{L^2(\Rsix)}&\le C \|H_\epsilon \psi^{(n)}_\epsilon - \mu^{(n)}_\epsilon \psi^{(n)}_\epsilon \|_{L^2(\Rsix)}
  \label{eq:lambda_eps2}
  \end{align}
  (the first estimate above follows from the Kato-Temple inequality~\cite{Kato}). To proceed further, we need to evaluate the $L^2$-norm of the residual $r_\epsilon^{(n)}:=H_\epsilon \psi^{(n)}_\epsilon - \mu^{(n)}_\epsilon \psi^{(n)}_\epsilon$.
 We have
$$
      H_\epsilon \pe^{(n)}
      = \me^{(n)} H_\epsilon \Te (\phi_\epsilon^{(n)}) =  \me^{(n)} \Te \left[ (H_0 + \Ve) \phi_\epsilon^{(n)})  \right]  = \me^{(n)} \Te \left[ (H_0 + \Ve) (\phi_0 + \sum_{k=3}^n \epsilon^{k} \phi_k) \right],
   $$
  and thus,
  \begin{equation*}
    \begin{split}
      r^{(n)}_\epsilon
      &= \me^{(n)} \Te \left[
        (H_0 + \Ve) \phi_\epsilon^{(n)}
        - \mu_\epsilon^{(n)} \phi_\epsilon^{(n)}
      \right] \\
      &= \me^{(n)} \Te \left[
        (H_0 + \Ve) (\phi_0 + \sum_{k=3}^n \epsilon^{k}\phi_k)
        - (\lambda_0 - \sum_{k=3}^n C_k \epsilon^{k} ) (\phi_0 + \sum_{k=3}^n \epsilon^{k}\phi_k)
        + (\lambda_\epsilon^{(n)} - \mu_\epsilon^{(n)})\phi_\epsilon^{(n)}
      \right] \\
      &= \me^{(n)} \Te \left[
        \left(H_0 + \sum_{k=3}^n \epsilon^{k} \mathcal{B}^{(k)} \right)
            (\phi_0 + \sum_{k=3}^n \epsilon^{k}\phi_k)
            - (\lambda_0 - \sum_{k=3}^n C_k \epsilon^{k} ) (\phi_0 + \sum_{k=3}^n \epsilon^{k}\phi_k)
        \right. \\ &\quad \left. \qquad\qquad
        + (\lambda_\epsilon^{(n)} - \mu_\epsilon^{(n)})\phi_\epsilon^{(n)} +
        (\Ve - \sum_{k=3}^n \epsilon^{k} \mathcal{B}^{(k)}) \phi_\epsilon^{(n)}
      \right].
    \end{split}
  \end{equation*}
Using~\eqref{eqn:PsiNEqnAnyOrderH0}, we get
  \begin{equation}\label{eqn:thmtrophmoleproof1}
    \begin{split}
      &(H_0 + \sum_{k=3}^n \epsilon^{k} \mathcal{B}^{(k)})(\phi_0 + \sum_{k=3}^n \epsilon^{k}\phi_k)
        - (\lambda_0 - \sum_{k=3}^n C_k \epsilon^{k} ) (\phi_0 + \sum_{k=3}^n \epsilon^{k}\phi_k) \\
      & \qquad \qquad = \epsilon^{n} \sum_{k=1}^{n} \epsilon^{k}
        \left(
          \sum_{j=k}^n \mathcal{B}^{(j)} \phi_{n+k-j} + \sum_{j=k}^n C_j \phi_{n+k-j}
        \right).
    \end{split}
  \end{equation}
Since  $\mathcal{B}^{(j)}$ are degree $(j-1)$ homogeneous functions (in cartesian coordinates) and the $\phi_n$'s decay exponentially in the sense of~\eqref{eq:exp_dec_phin}, there exists $K_n \in \R_+$ and $\epsilon_n > 0$ such that for all $0 < \epsilon \le \epsilon_n$,
  \begin{equation} \label{eq:res1}
    \norm{  (H_0 + \sum_{k=3}^n \epsilon^{k} \mathcal{B}^{(k)})(\phi_0 + \sum_{k=3}^n \epsilon^{k}\phi_k)
        - (\lambda_0 - \sum_{k=3}^n C_k \epsilon^{k} ) (\phi_0 + \sum_{k=3}^n \epsilon^{k}\phi_k)
      }_{L^2(\Rsix)}
      \le K_n \epsilon^{n+1}.
  \end{equation}
  It remains to bound $\| (\Ve - \sum_{k=3}^n \epsilon^{k} \mathcal{B}^{(k)}) \psi_\epsilon^{(n)}\|_{L^2(\Rsix)}$. From \eqref{eqn:limsighr}, \eqref{eq:exp_dec_phin} and \eqref{eqn:normltestfn}, there exists $\epsilon_n > 0$, $\alpha_n > 0$ and $M_n \in \R_+$ such that for all $0 < \epsilon \le \epsilon_n$
$$
\|e^{ \alpha_n (|\br_1|+|\br_2|)} \phi_\epsilon^{(n)}\|_{H^1(\Rsix)} \le M_n.
$$
Introducing
\begin{equation} \label{eqn:omegaepsilon}
\Omega_\epsilon=\set{(\rr_1,\rr_2)\in\Rsix}{|\rr_1|+|\rr_2| < (2\epsilon)^{-1}}.
\end{equation}
and the potentials defined by
\begin{equation} \label{eqn:eyesplitdef}
\begin{split}
 v^{(1)}_\epsilon(\rr_1,\rr_2)&:= |\rr_1 -\epsilon^{-1} \be|^{-1},\quad
 v^{(2)}_\epsilon(\rr_1,\rr_2) := |\rr_2+\epsilon^{-1}\be|^{-1}, \quad
 v^{(3)}_\epsilon(\rr_1,\rr_2):= |\rr_1 -\rr_2-\epsilon^{-1} \be|^{-1},
\end{split}
\end{equation}
we have,
\begin{align*}
\| (\Ve - \sum_{k=3}^n \epsilon^{k} \mathcal{B}^{(k)}) \phi_\epsilon^{(n)} \|_{L^2(\Rsix)} & \le \| (\Ve - \sum_{k=3}^n \epsilon^{k} \mathcal{B}^{(k)}) \phi_\epsilon^{(n)} \|_{L^2(\Omega_\epsilon)}
+  \sum_{k=3}^n \epsilon^{k} \|  \mathcal{B}^{(k)}  \phi_\epsilon^{(n)}\|_{L^2(\Omega_\epsilon^c)} \\
& \qquad + \sum_{j=1}^3 \|  v^{(j)}_\epsilon \phi_\epsilon^{(n)}\|_{L^2(\Omega_\epsilon^c)}  + \epsilon \| \phi_\epsilon^{(n)}\|_{L^2(\Omega_\epsilon^c)}.
\end{align*}
We first see that
$$
 \| \phi_\epsilon^{(n)}\|_{L^2(\Omega_\epsilon^c)} \le e^{-\alpha_n(2\epsilon)^{-1}}  \| e^{\alpha_n (|\br_1|+|\br_2|)}  \phi_\epsilon^{(n)}\|_{L^2(\Omega_\epsilon^c)} \le M_n  e^{-\alpha_n(2\epsilon)^{-1}}.
$$
Next, as $\mathcal{B}^{(k)}$ is a polynomial function, there exists a constant $B_n$ such as for all $0 < \epsilon \le \epsilon_n$,
\begin{align*}
 \sum_{k=3}^n \epsilon^{k} \| \mathcal{B}^{(k)}  \phi_\epsilon^{(n)}\|_{L^2(\Omega_\epsilon^c)}
& \le \sum_{k=3}^n \epsilon^{k} \| \mathcal{B}^{(k)} e^{-\alpha_n (|\br_1|+|\br_2|)} \|_{L^\infty(\Omega_\epsilon^c)} \| e^{\alpha_n (|\br_1|+|\br_2|)} \phi_\epsilon^{(n)}\|_{L^2(\Omega_\epsilon^c)}  \\
& \le M_n  \sum_{k=3}^n \epsilon^{k} \| \mathcal{B}^{(k)} e^{-\alpha_n (|\br_1|+|\br_2|)} \|_{L^\infty(\Omega_\epsilon^c)} \le B_n \epsilon^3 e^{-\alpha_n (2\epsilon)^{-1}}.
\end{align*}
In addition, we have
\begin{align*}
 \sum_{j=1}^3 \| v^{(j)}_\epsilon \phi_\epsilon^{(n)}\|_{L^2(\Omega_\epsilon^c)} &\le   \sum_{j=1}^3 e^{-\alpha_n (2\epsilon)^{-1}} \| v^{(j)}_\epsilon e^{\alpha_n (|\br_1|+|\br_2|)} \phi_\epsilon^{(n)}\|_{L^2(\Omega_\epsilon^c)}\\
 & \le  \sum_{j=1}^3  e^{-\alpha_n (2\epsilon)^{-1}} \| v^{(j)}_\epsilon e^{\alpha_n (|\br_1|+|\br_2|)} \phi_\epsilon^{(n)}\|_{L^2(\Rsix)} \\
 & \le 8 e^{-\alpha_n (2\epsilon)^{-1}} \| e^{\alpha_n (|\br_1|+|\br_2|)} \phi_\epsilon^{(n)}\|_{H^1(\Rsix)}= 8 e^{-\alpha_n (2\epsilon)^{-1}} M_n,
\end{align*}
where we have used the Hardy inequality in dimension 3
$$
\forall \phi \in H^1(\R^3), \quad \int_{\R^3} \frac{|\phi(\br)|^2}{|\br|^2} \, d\br \le 4 \int_{\R^3} |\nabla \phi(\br)|^2 \, d\br
$$
to show that for any $\psi \in H^1(\Rsix)$,
\begin{align*}
\|v^{(j)}_\epsilon \psi\|_{L^2(\Rsix)}^2 &= \int_{\R^3} \left( \int_{\R^3} \frac{|\psi(\br_1,\br_2)|^2}{|\br_j+(-1)^j\epsilon^{-1}\be|^2} \, d\br_j \right) \, d\br_{3-j} \\
&\le  \int_{\R^3} 4 \left( \int_{\R^3} |\nabla_{\br_j} \psi(\br_1,\br_2)|^2 \, d\br_j \right) \, d\br_{3-j} \le 4 \| \nabla_{\br_j} \psi\|_{L^2(\Rsix)}^2,
\end{align*}
for $j=1,2$, and
\begin{align*}
\|v^{(3)}_\epsilon \psi\|_{L^2(\Rsix)}^2 &= \int_{\R^3} \int_{\R^3} \frac{|\psi(\br_1,\br_2)|^2}{|\br_1-\br_2-\epsilon^{-1}\be|^2} \, d\br_1\, d\br_2
= \frac 18 \int_{\R^3} \int_{\R^3} \frac{|\psi\left( \br_1'+\br_2',\br_1'-\br_2' \right)|^2}{|\br_2'-\epsilon^{-1}\be|^2} \, d\br_1' \, d\br_2' \\
&\le \frac 12 \int_{\R^3} \int_{\R^3} |(\nabla_{\br_1}-\nabla_{\br_2})\psi\left( \br_1'+\br_2',\br_1'-\br_2' \right)|^2 \, d\br_1' \, d\br_2' \\
=& 4 \|(\nabla_{\br_1}-\nabla_{\br_2})\psi\|_{L^2(\Rsix)}^2 = 8  \|\nabla\psi\|_{L^2(\Rsix)}^2.
\end{align*}
From the multipolar expansion of $V_\epsilon$, we know that there exist $c_n \in \R_+$
\begin{equation}
\label{eqn:VEpsilonExpansion}
\left| \Ve(\mathbf{r}_1, \mathbf{r}_2) - \sum_{i=3}^n \epsilon^i
       \mathcal{B}^{(i)}(\mathbf{r}_1,\mathbf{r}_2)
\right| \le c_n K^n \epsilon^{n+1}, \quad\hbox{whenever}\; |\br_1|+|\br_2| \le K\leq (2\epsilon)^{-1}.
\end{equation}
Let us now show that \eqref{eqn:VEpsilonExpansion} implies that there exists $\widetilde c_{n} \in \R_+$ such that for all $0 \le K\leq (2\epsilon)^{-1}$,
\begin{equation}\label{eqn:locestbeexp}
\sup_{|\rr_1|+|\rr_2|\leq K} \Big|V_\epsilon(\rr_1,\rr_2)
 -\sum_{i=3}^n \epsilon^i {\cal B}^{(i)}(\rr_1,\rr_2)\Big|e^{-\alpha_n(|\rr_1|+|\rr_2|)}
\leq \widetilde c_{n} \epsilon^{n+1},
\end{equation}
This is immediate from \eqref{eqn:VEpsilonExpansion} for $K\leq 1$, taking $\widetilde c_{n}=c_n$.
Now we let $K> 1$.
Then \eqref{eqn:VEpsilonExpansion} implies
$$
\sup_{(K/2)\leq(|\rr_1|+|\rr_2|)\leq K} \Big|V_\epsilon(\rr_1,\rr_2)
 -\sum_{i=3}^n \epsilon^i {\cal B}^{(i)}(\rr_1,\rr_2)\Big|e^{-\alpha_n(|\rr_1|+|\rr_2|)}
\leq c_n e^{-\alpha_n K/2} K^n \epsilon^{n+1}.
$$
Applying this repeatedly for $2^{-j}K$ replacing $K$ until $2^{-j}K<1$ yields
\eqref{eqn:locestbeexp}, with
$$
\widetilde c_{n} = c_n \sup_{t \ge 0}  t^n e^{-\alpha_n t/2}.
$$
Applying \eqref{eqn:locestbeexp} for $K=(2\epsilon)^{-1}$ yields
$$
 \| (\Ve - \sum_{k=3}^n \epsilon^{k} \mathcal{B}^{(k)}) e^{-\alpha_n(|\rr_1|+|\rr_2|)}\|_{L^\infty(\Omega_\epsilon)}
\leq \widetilde  c_{n} \epsilon^{n+1},
$$
from which we obtain
\begin{align*}
 \| (\Ve - \sum_{k=3}^n \epsilon^{k} \mathcal{B}^{(k)}) \phi_\epsilon^{(n)}\|_{L^2(\Omega_\epsilon)} &
 \le  \| (\Ve - \sum_{k=3}^n \epsilon^{k} \mathcal{B}^{(k)}) e^{-\alpha_n(|\rr_1|+|\rr_2|)}\|_{L^\infty(\Omega_\epsilon)}   \|e^{\alpha_n(|\rr_1|+|\rr_2|)}\phi_\epsilon^{(n)}\|_{L^2(\Omega_\epsilon)} \\
 & \le \widetilde c_n M_n \epsilon^{n+1}.
\end{align*}
Finally, we get
\begin{equation}\label{eq:Ve-ME}
\| (\Ve - \sum_{k=3}^n \epsilon^{k} \mathcal{B}^{(k)}) \phi_\epsilon^{(n)}\|_{L^2(\Rsix)} \le
 \widetilde c_n M_n \epsilon^{n+1} + (8+\epsilon+ B_n \epsilon^3) M_n e^{-\alpha_n (2\epsilon)^{-1}},
 \end{equation}
 Together with \eqref{eq:res1}, this proves that there exists $c''_n \in \R_+$ such that for all $0 < \epsilon \le \epsilon_n$,
 \begin{equation}\label{eq:bound_residual_L2}
 \|r_\epsilon^{(n)}\|_{L^2(\Rsix)} = \| H_\epsilon \psi^{(n)}_\epsilon - \mu^{(n)}_\epsilon \psi^{(n)}_\epsilon\|_{L^2(\Rsix)} \le c''_n  \epsilon^{n+1}.
\end{equation}
It follows from \eqref{eq:lambda_eps}-\eqref{eq:lambda_eps2} that for $n \ge 3$ fixed, there exists $C \in \R_+$ such that for all $0 < \epsilon \le \epsilon_n$,
  \begin{equation}\label{eq:Wigner2n+1}
  |\lambda_\epsilon - \mu^{(n)}_\epsilon |  \le C \epsilon^{2(n+1)} \quad \mbox{and} \quad
  \|  \psi_\epsilon - \pe^{(n)} \|_{L^2(\Rsix)}\le C \epsilon^{n+1}.
  \end{equation}
 Then,
  \begin{align*}
      \mu_\epsilon^{(n)} - \lambda_\epsilon^{(n)}
      &= \inner{\pe^{(n)}}{H_\epsilon \pe^{(n)} - \lambda_\epsilon^{(n)} \pe^{(n)}}  \\
      &= \me^{(n)}  \bigg\langle \pe^{(n)},
        \Te \Big[
          (\Ve  - \sum_{k=3}^n \epsilon^{k} \mathcal{B}^{(k)})
          \phi_\epsilon^{(n)}
  +
  \epsilon^{n} \sum_{k=1}^{n} \epsilon^{k}
    \bigg(
      \sum_{j=k}^n \mathcal{B}^{(j)} \phi_{n+k-j} + \sum_{j=k}^n C_j \phi_{n+k-j}
    \bigg)
        \Big]
\bigg\rangle
  \end{align*}
  so that there exists a constant $c_n$ such that for $0 < \epsilon \le  \epsilon_n$,
  \begin{align*}
      \left|\mu_\epsilon^{(n)} - \lambda_\epsilon^{(n)} \right|
      &\leq 2  \biggnorm{(\Ve  - \sum_{k=3}^n \epsilon^{k} \mathcal{B}^{(k)})
        \phi_\epsilon^{(n)}
        +
        \epsilon^{n} \sum_{k=1}^{n} \epsilon^{k}
          \bigg(
            \sum_{j=k}^n \mathcal{B}^{(j)} \phi_{n+k-j} + \sum_{j=k}^n C_j \phi_{n+k-j}
          \bigg)
        }_{L^2(\Rsix)} \\
      &\leq c_n \epsilon^{n+1}.
  \end{align*}
The error bounds on the eigenvalue errors in \eqref{eqn:finalestionsh} follow from \eqref{eq:Wigner2n+1} and the above inequality.

\medskip

Finally, the error $\xi_\epsilon^{(n)} = \psi_\epsilon - \psi_\epsilon^{(n)}$, as defined in~\cite{lrsBIBhi},
satisfies
\[
  H_\epsilon \xi_\epsilon^{(n)} = \lambda_\epsilon \psi_\epsilon - H_\epsilon \psi_\epsilon^{(n)} = \lambda_\epsilon - \mu_\epsilon^{(n)} - r_\epsilon^{(n)} =: \eta_\epsilon^{(n)}.
\]
From \eqref{eq:bound_residual_L2}-\eqref{eq:Wigner2n+1}, there exists a constant $c_n \in \R_+$ such that for all $0 < \epsilon \le \epsilon_n$,
$$
\| \xi_\epsilon^{(n)} \|_{L^2(\Rsix)} \le c_n \epsilon^{n+1} \quad \mbox{and} \quad
\| \eta_\epsilon^{(n)} \|_{L^2(\Rsix)} \le c_n \epsilon^{n+1}.
$$
In addition,
\begin{equation}\label{eq:last_bounds}
  - \frac 12 \Delta \xi_\epsilon^{(n)} = - W_\epsilon \xi_\epsilon^{(n)} + \eta_\epsilon^{(n)},
\end{equation}
where
$$
W_\epsilon(\br_1,\br_2):= -\frac{1}{|{\bR_1}-(2\eps)^{-1} \be|}- \frac{1}{|{\bR_2}-(2\eps)^{-1}  \be|}
-\frac{1}{|{\bR_1}+(2\eps)^{-1}  \be|}- \frac{1}{|{\bR_2}+(2\eps)^{-1}  \be|}
+\frac{1}{|{\bR_1}-{\bR_2}|}+\eps.
$$
Proceeding as in~\cite[Section 2.4]{lrsBIBhi}, we use the Hardy inequality in $\R^3$ and the Cauchy-Schwarz inequality to obtain that
\begin{align*}
\frac 12 \|\nabla \xi_\epsilon^{(n)}\|_{L^2(\R^3 \times \R^3)}^2 &= \langle  \xi_\epsilon^{(n)}, - W_\epsilon \xi_\epsilon^{(n)} + \eta_\epsilon^{(n)} \rangle \\
&\le (10 \|\nabla \xi_\epsilon^{(n)}\|_{L^2(\R^3 \times \R^3)} + \epsilon  \| \xi_\epsilon^{(n)}\|_{L^2(\R^3 \times \R^3)}+ \| \eta_\epsilon^{(n)}\|_{L^2(\R^3 \times \R^3)})  \| \xi_\epsilon^{(n)}\|_{L^2(\R^3 \times \R^3)},
\\
 \frac 12 \|\Delta  \xi_\epsilon^{(n)}\|_{L^2(\R^3 \times \R^3)} &= \| - W_\epsilon \xi_\epsilon^{(n)} + \eta_\epsilon^{(n)}\|_{L^2(\R^3 \times \R^3)} \\
&\le 10  \|\nabla \xi_\epsilon^{(n)}\|_{L^2(\R^3 \times \R^3)} +  \epsilon \|\xi_\epsilon^{(n)}\|_{L^2(\R^3 \times \R^3)} +  \|\eta_\epsilon^{(n)}\|_{L^2(\R^3 \times \R^3)}.
\end{align*}
It follows from \eqref{eq:last_bounds} that there exists a constant $c_n \in \R_+$ such that for all $0 < \epsilon \le \epsilon_n$, $\| \Delta \xi_\epsilon^{(n)} \|_{L^2(\Rsix)} \le c_n \epsilon^{n+1}$, and thus $\| \xi_\epsilon^{(n)} \|_{H^2(\Rsix)} \le c_n \epsilon^{n+1}$.

\section*{Acknowledgements}

EC would like to thank the Department of Mathematics at the
University of Chicago for generous support.
This work has received funding from the European Research Council (ERC) under the
European Union's Horizon 2020 research and innovation programme
(grant agreement No 810367).
The Labex B\'ezout is acknowledged for funding the PhD thesis of Rafa\"el Coyaud.
We are grateful to Virginie Ehrlacher for useful comments on the manuscript.

\appendix

\section{Appendix}

\subsection{Multipolar expansion of $V_\eps$}
\label{sect:RSH}

We start from the well-known multipolar expansion of $\frac{1}{|\mathbf{r}-R \mathbf{e}|}$ in terms
of Legendre polynomials
\begin{equation}
 \frac{1}{| \mathbf{r}-R \mathbf{e}|} = \frac{1}{R}
\left(\sum_{k = 0}^\infty P_k\Big(\frac{\mathbf{r}\cdot\mathbf{e}}{|\mathbf{r}|}\Big)
    \left( \frac{|\mathbf{r}|}{R} \right)^k \right), \quad \text{for} \, |\mathbf{r}| < R, \label{eq:expLeg}
\end{equation}
which is a straightforward consequence of the definition of Legendre polynomials via their generating function~\cite{wan2013generating}
\begin{equation}\label{eqn:genfunvalid}
\forall -1 \le x \le 1, \quad \big(1-2xt+t^2\big)^{-1/2}=\sum_{k = 0}^\infty P_k(x)t^k,
\end{equation}
taking
$$
-1 \le x=\frac{\mathbf{r}\cdot\mathbf{e}}{|\mathbf{r}|} \le 1,\qquad
t=\frac{|\mathbf{r}|}{R}.
$$
Since the Legendre polynomials are at most 1 in magnitude on the interval $[-1,1]$, the sum in
\eqref{eqn:genfunvalid} converges absolutely for all $|t|<1$, and
$$
\Big|\sum_{k = n}^\infty P_k(x)t^k\Big|\leq \sum_{k = n}^\infty t^k
=\frac{t^n}{1-t}\leq 2 t^n,
\quad\hbox{for all} \; |t| \le \half.
$$
Consequently,
\begin{equation}\label{eqn:isthisvalid}
\bigg| \frac{1}{| \mathbf{r}-R \mathbf{e}|} -  \frac{1}{R}
\left(\sum_{k = 0}^{n-1} P_k\Big(\frac{\mathbf{r}\cdot\mathbf{e}}{|\mathbf{r}|}\Big)
                  \left( \frac{|\mathbf{r}|}{R} \right)^k \right)\bigg|
\leq 2 \frac{|\mathbf{r}|^{n}}{R^{n+1}},
\quad\hbox{for all} \; |\mathbf{r}| \le R/2.
\end{equation}
Recalling that $P_0(x)=1$, $P_1 (x)  = x$ and
$$
V_\epsilon({\br_1},{\br_2})=- \frac{1}{|{\br_1}-\epsilon^{-1}\be|}- \frac{1}{|{\br_2}+\epsilon^{-1}\be|}+\frac{1}{|{\br_1}-{\br_2}-\epsilon^{-1}\be|}+\epsilon.
$$
with $\epsilon=R^{-1}$, we deduce from \eqref{eqn:isthisvalid} that
\begin{equation}\label{eq:boundVeps}
\left| \Ve(\mathbf{r}_1, \mathbf{r}_2) - \sum_{k=3}^n \epsilon^k
       \mathcal{B}^{(k)}(\mathbf{r}_1,\mathbf{r}_2)
\right| \le 6 K^n \epsilon^{n+1}, \quad\hbox{whenever}\; |\br_1|+|\br_2| \le K\leq (2\epsilon)^{-1},
\end{equation}
where the polynomial functions $\mathcal{B}^{(k)}$ are given by
$$
\mathcal{B}^{(k)}(\mathbf{r}_1,\mathbf{r}_2):= P_{k-1}\left(\frac{(\br_1-\br_2)\cdot\be}{|\br_1-\br_2|} \ \right) |\br_1-\br_2|^{k-1}
- P_{k-1}\left(\frac{\br_1\cdot\be}{|\br_1|} \ \right) |\br_1|^{k-1} - P_{k-1}\left(-\frac{\br_2\cdot\be}{|\br_2|} \ \right) |\br_2|^{k-1}.
$$
This proves \eqref{eqn:VEpsilonExpansion}. To derive the expression \eqref{eqnBic} for the $\mathcal{B}^{(k)}$'s, we first use the identities
\begin{align*}
&P_l(\sigma \cdot \sigma') = \left(\frac{4 \pi}{2l+1}\right) \sum_{m = -l}^l (-1)^m Y_l^m(\sigma) Y_l^m(\sigma'), \qquad
      \sqrt{\frac{4 \pi}{2l+1}} Y_l^{m}(\be) = \delta_{m,0},
\end{align*}
valid for all $l \in \N$, $- l \le m \le l$, $\sigma,\sigma' \in {\mathbb S}^2$ (recall that $\be$ is the unit vector of the $z$-axis), and get
$$
\mathcal{B}^{(k)}(\mathbf{r}_1,\mathbf{r}_2):= \sqrt{\frac{4 \pi}{2k-1}}  \left(Y_{k-1}^0\left(\frac{\br_1-\br_2}{|\br_1-\br_2|} \ \right) |\br_1-\br_2|^{k-1}
- Y_{k-1}^0\left(\frac{\br_1}{|\br_1|} \ \right) |\br_1|^{k-1} - Y_{k-1}^0\left(-\frac{\br_2}{|\br_2|} \ \right) |\br_2|^{k-1}\right).
$$
We next use the addition formula~\cite{ToughStone} stating that for $l \in \N$, $\br_1, \br_2 \in \R^3$,
\begin{align*}
 \sqrt{\frac{4\pi}{2l+1}}Y_l^0\left(\frac{\br_1-\br_2}{|\br_1-\br_2|}\right) |\mathbf{r}_1 - \mathbf{r}_2|^l
= \sum_{l_1+l_2=l} \sum_{m=-\min(l_1,l_2)}^{\min(l_1,l_2)} G_{\rm c}(l_1, l_2, m)
     r_1^{l_1} Y_{l_1}^{m}\left(\frac{\br_1}{|\br_1|}\right) r_2^{l_2} Y_{l_2}^{-m}\left(\frac{\br_2}{|\br_2|}\right),
     \end{align*}
 where
\begin{align*}
 G_{\rm c}(l_1, l_2, m) &= (-1)^{l_2} \frac{4 \pi}{((2l_1 + 1)(2l_2 + 1))^{1/2}} \begin{pmatrix} l_1 + l_2  \\ l_1 + m \end{pmatrix}^{1/2}
 \begin{pmatrix} l_1+l_2 \\ l_1 - m\end{pmatrix}^{1/2}, \\
 &= (-1)^{l_2} \frac{4 \pi (l_1 + l_2)!}{((2l_1 + 1)(2l_2 + 1)(l_1 + m)! (l_2 +m)! (l_1 - m)! (l_2 - m)!)^{1/2}}.
\end{align*}
As for $G_{\rm c}(l, 0, 0)=G_{\rm c}(0, l, 0)=\frac{4\pi}{(2l+1)^{1/2}}$ and $Y_0^0=\frac{1}{\sqrt{4\pi}}$, we finally obtain  \eqref{eqnBic}.

\section{Wigner $(2n+1)$ rule}

Using the notation in \eqref{eqn:normltestfn}, we consider the Rayleigh quotients
$$
\mu_\epsilon^{(n)} = \langle \psi_\epsilon^{(n)}, H_\epsilon \psi_\epsilon^{(n)} \rangle \quad \mbox{and} \quad
\widetilde \mu_\epsilon^{(n)} = \frac{\inner{\phi_\epsilon^{(n)}}{\left(   H_0 + \sum_{i=3}^{2n+1} \epsilon^{i} \mathcal{B}^{(i)} \right) \phi_\epsilon^{(n)}}}{\|\phi_\epsilon^{(n)}\|_{L^2(\Rsix)}^2}
$$
(recall that $\|\psi_\epsilon^{(n)}\|_{L^2(\Rsix)}=1$). Let
$$
\eta_\epsilon^{(n)}:=(H_0+V_\epsilon)\phi_\epsilon^{(n)}, \quad  \upsilon_\epsilon^{(n)}:= (V_\epsilon-\sum_{i=3}^{2n+1} \epsilon^{i} \mathcal{B}^{(i)})\phi_\epsilon^{(n)}\quad \mbox{and} \quad \xi_\epsilon^{(n)}:=(\Te^*\Te-1)\phi_\epsilon^{(n)}.
$$
We deduce from the boundedness of the $\phi_n$'s in $H^2(\Rsix)$, the Hardy inequality in $\R^3$, and the estimates \eqref{eq:exp_dec_phin} and \eqref{eqn:VEpsilonExpansion}, that there exist $C \in \R_+$, $\beta_n > 0$ and $\epsilon_n > 0$ such that for all $0 \le  \epsilon \le \epsilon_n$
$$
\|\phi_\epsilon^{(n)}\|_{L^2(\Rsix)} \le 2, \quad \|\eta_\epsilon^{(n)}\|_{L^2(\Rsix)} \le C, \quad \|\upsilon_\epsilon^{(n)}\|_{L^2(\Rsix)} \le C \epsilon^{2n+2}, \quad
\|\xi_\epsilon^{(n)}\|_{L^2(\Rsix)} \le C e^{-\beta_n\epsilon},
$$
proceeding as in the proof of \eqref{eq:Ve-ME} to establish the third inequality. It follows from \eqref{eqn:finalestionsh} and the above bounds that
\begin{align*}
  \widetilde \mu_\epsilon^{(n)}  &=\lambda_\epsilon +  \widetilde \mu_\epsilon^{(n)} - \mu_\epsilon^{(n)} + O(\epsilon^{2n+2}) \\
&=\lambda_\epsilon  + \frac{\inner{\phi_\epsilon^{(n)}}{\left(   H_0 + \sum_{i=3}^{2n+1} \epsilon^{i} \mathcal{B}^{(i)} \right) \phi_\epsilon^{(n)}}}{\|\phi_\epsilon^{(n)}\|_{L^2(\Rsix)}^2}- \frac{\inner{\Te^*\Te\phi_\epsilon^{(n)}}{(H_0+V_\epsilon)\phi_\epsilon^{(n)}}}{\inner{\Te^*\Te\phi_\epsilon^{(n)}}{\phi_\epsilon^{(n)}}} + O(\epsilon^{2n+2}) \\
&= \lambda_\epsilon -  \frac{\inner{\phi_\epsilon^{(n)}}{\upsilon_\epsilon^{(n)}}}{\inner{\phi_\epsilon^{(n)}}{\phi_\epsilon^{(n)}}} + \frac{\inner{\xi_\epsilon^{(n)}}{\eta_\epsilon^{(n)}}-\inner{\xi_\epsilon^{(n)}}{\phi_\epsilon^{(n)}}\inner{\phi_\epsilon^{(n)}}{\eta_\epsilon^{(n)}}}{\inner{\phi_\epsilon^{(n)}}{\phi_\epsilon^{(n)}}+\inner{\xi_\epsilon^{(n)}}{\phi_\epsilon^{(n)}}}+ O(\epsilon^{2n+2}) \\
&= \lambda_\epsilon + O(\epsilon^{2n+2}) = -1 - \sum_{k=6}^{2n+1}  C_k \epsilon^k + O(\epsilon^{2n+2}).
\end{align*}
Thus, the coefficients $C_k$ for $k \le 2n+1$ can be computed from the Taylor expansion of $\widetilde \mu_\epsilon^{(n)}$ up to order $(2n+1)$, which only involves the $\phi_k$'s for $k \le n$, and the $\mathcal{B}^{(k)}$'s for $k \le (2n+1)$. To obtain a computable expression of the coefficients $C_{2n}$ and $C_{2n+1}$, we first use Equation \eqref{eqn:PsiNEqnAnyOrderH0}, which can be rewritten as
\begin{equation}
  H_0 \phi_{k} + \sum_{j=3}^k \mathcal{B}^{(j)}\phi_{k-j}
  = - C_0 \phi_{k} - \sum_{j=6}^k C_j \phi_{k-j} = - \sum_{j=0}^k C_j \phi_{k-j},
\end{equation}
with $C_0 =1$ and $C_i = 0$ for $i = 1, ..., 5$, to get that for all $n \ge 1$
\begin{align}
\nu_\epsilon^{(n)}:&=\inner{\phi_\epsilon^{(n)}}{\left(   H_0 + \sum_{i=3}^{2n+1} \epsilon^{i} \mathcal{B}^{(i)} \right) \phi_\epsilon^{(n)}}  \nonumber \\
&=  - \sum_{l=0}^n \epsilon^{l} \sum_{i=0}^l \inner{\phi_i}{\sum_{j=0}^{l-i} C_j \phi_{l-i-j}}
    + \epsilon^n \sum_{l = 1}^n \epsilon^l \left( - \sum_{i=l}^n \inner{\phi_i}{\sum_{j=0}^{n+l-i} C_j \phi_{n+l-i-j}}
    + \sum_{i=0}^{l-1} \inner{\phi_i}{\sum_{j=0}^n {\cal B}^{(n + l-i - j)} \phi_{j}} \right)  \nonumber \\
    &\quad + \epsilon^{2n+1}   \sum_{i=0}^n \inner{\phi_i}{\sum_{j=0}^n {\cal B}^{(2n + 1 - i - j)}\phi_{j}} + O(\epsilon^{2n+2}). \label{eqn:tildepsien}
\end{align}
In addition, we have
$$
    \| \phi_\epsilon^{(n)} \|^2
    = \inner{\sum_{i=0}^n \epsilon^{i} \phi_{i}}{\sum_{i=0}^n \epsilon^{i} \phi_{j}}
    = 1 + \sum_{k=1}^n \epsilon^{k} \sum_{i=0}^{k} \inner{\phi_{i}}{\phi_{k-i}}
    + \epsilon^{n} \sum_{k=1}^n \epsilon^{k} \sum_{i=k}^n \inner{\phi_{i}}{\phi_{n+k-i}},
$$
and, using the relation $\sum_{i=0}^{k} \inner{\phi_{i}}{\phi_{k-i}} = 0$ derived from \eqref{eqn:PsiNEqnAnyOrderPsi0}, we get
\begin{equation}
  \| \phi_\epsilon^{(n)} \|^2
  = 1
  + \epsilon^{n} \sum_{k=1}^n \epsilon^{k} \sum_{i=k}^n \inner{\phi_{i}}{\phi_{n+k-i}}. \label{eqn:tildepsinorm}
\end{equation}
Il follows from \eqref{eqn:tildepsien}-\eqref{eqn:tildepsinorm} that
$$
\widetilde \mu_\epsilon^{(n)} = \frac{\nu_\epsilon^{(n)}}{\| \phi_\epsilon^{(n)} \|^2} =- \sum_{k=0}^{2n+1}  C_k \epsilon^k + O(\epsilon^{2n+2}),
$$
with
$$
   C_{2n} = \inner{\phi_n}{\sum_{j=0}^{n} C_j \phi_{n-j}}
  - \sum_{i=0}^{n-1} \inner{\phi_i}{\sum_{j=0}^n {\cal B}^{(2n -i - j)} \phi_{j}}
  - \sum_{k=1}^n \left(\sum_{i=k}^n \inner{\phi_{i}}{\phi_{n+k-i}}\right) \sum_{i=0}^{n-k} \inner{\phi_i}{\sum_{j=0}^{n-k-i} C_j \phi_{n-k-i-j}},
$$
and
$$
   C_{2n+1} = - \sum_{i=0}^n \inner{\phi_i}{\sum_{j=0}^n {\cal B}^{(2n + 1 - i - j)}\phi_{j}}
  - \sum_{k=1}^n \left(\sum_{i=k}^n \inner{\phi_{i}}{\phi_{n+k-i}}\right) \sum_{i=0}^{n+1-k} \inner{\phi_i}{\sum_{j=0}^{n+1-k-i} C_j \phi_{n+1-k-i-j}}.
$$

\section{Computation of the integrals $S_n$ in \eqref{eq:Sn}}
\label{sec:Sn}

Recall that
$$
S_n = \int_{0}^{+ \infty} r^3 e^{-r} \varphi_{n,1}(r) \dd r,
$$
where
$$
 \varphi_{n,1} =  \sqrt{\left(\frac{2}{n}\right)^3 \frac{(n-2)!}{2n(n+1)!}}
  \left(\frac{2r}{n}\right) L_{n-2}^{(3)}\left(\frac{2r}{n}\right) e^{-r/n},
$$
where the associated Laguerre polynomials of the second kind $L_n^{(m)}$, $n,m \in \N$, satisfy the following properties~\cite[Section 22.2]{abramowitz1970handbook}:
\begin{itemize}
\item for all $k,k',m \in \N$,
 \begin{equation} \label{eq:ortho_Lnm}
    \int_0^\infty x^m L_{k}^{(m)}(x) L_{k'}^{(m)}(x) e^{-x} \, \dd x = \frac{(k+m)!}{k!}\delta_{k,k'};
  \end{equation}
\item for all $\gamma \in \C$ such that $\Re(\gamma) > -\frac 12$, and $m \in \N$,
\begin{equation} \label{eq:dev_exp_Lnm}
    e^{-\gamma  x} = \sum_{k=0}^{+\infty} \frac{\gamma^k}{(1 + \gamma)^{k + m + 1}} L_k^{(m)} (x);
  \end{equation}
  \item for all $k,m \in \N$,
   \begin{equation} \label{eq:xLmn}
    x L_k^{(m+1)}(x) = (k + m+1)L_k^{(m)}(x) - (k+1)L_{k+1}^{(m)}(x).
  \end{equation}
\end{itemize}
By a change of variable, we obtain
$$
S_n = \frac{n^2}8  \sqrt{\frac{(n-2)!}{(n+1)!}} I_n \quad \mbox{with} \quad I_n:=\int_0^{+\infty} x^4 L_{n-2}^{(3)} e^{-\frac{n-1}2 x} e^{-x} \, dx.
$$
Applying \eqref{eq:dev_exp_Lnm} for $\gamma=\frac{n-1}2$ and $m=4$, then \eqref{eq:xLmn} for $m=3$, and finally \eqref{eq:ortho_Lnm} for $m=3$, we obtain
\begin{align*}
I_n&= \int_0^{+\infty} x^4 L_{n-2}^{(3)} \left(  \sum_{k=0}^{+\infty} \frac{2^{5}(n-1)^k}{(n+1)^{k + 5}} L_k^{(4)} (x) \right) e^{-x} \, dx \\
&=  \int_0^{+\infty} x^3 L_{n-2}^{(3)} \left(  \sum_{k=0}^{+\infty} \frac{2^{5}(n-1)^k}{(n+1)^{k + 5}}
\left( (k +4)L_k^{(3)}(x) - (k+1)L_{k+1}^{(3)}(x) \right) \right) e^{-x} \, dx \\
&= \sum_{k=0}^{+\infty}  \frac{2^{5}(n-1)^k}{(n+1)^{k + 5}} \left( (k+4) \frac{(k+3)!}{k!}\delta_{k,n-2} - (k+1) \frac{(k+4)!}{(k+1)!}\delta_{k+1,n-2} \right) \\
&= \frac{2^{5}(n-1)^{n-2}}{(n+1)^{n + 3}}   (n+2) \frac{(n+1)!}{(n-2)!} -  \frac{2^{5}(n-1)^{n-3}}{(n+1)^{n + 2}} (n-2) \frac{(n+1)!}{(n-2)!} \\
&= \frac{2^{6}n(n-1)^{n-3}}{(n+1)^{n + 3}}  \frac{(n+1)!}{(n-2)!} .
\end{align*}
Finally, we get
$$
S_n = 8 n^3 \frac{(n-1)^{n-3}}{(n+1)^{n + 3}}   \sqrt{\frac{(n+1)!}{(n-2)!}} .
$$

\bibliography{newnext}
\bibliographystyle{abbrv}


\end{document}